\documentclass[journal,twocolumn,final]{IEEEtran}

\usepackage{amsmath,amssymb}
\usepackage{amsthm}
\usepackage{amsfonts}
\usepackage{mathtools}

\usepackage{color}

\usepackage[ breaklinks,
			 colorlinks = true,
             linkcolor = blue,
             urlcolor  = blue,
             citecolor = red,
             anchorcolor = blue,
]{hyperref}

\usepackage{enumitem}
\setlist[1]{itemsep=0pt}

\newtheorem{lemma}{Lemma}
\newtheorem{theorem}[lemma]{Theorem}
\newtheorem{prop}[lemma]{Proposition}
\newtheorem{corollary}[lemma]{Corollary}
\theoremstyle{plain}
\newtheorem{definition}{Definition}

\DeclareMathOperator{\supp}{supp}

\newcommand{\cP}{\mathcal{P}}
\newcommand{\cPp}{\mathcal{P}_{>0}}
\newcommand{\cS}{\mathcal{S}}
\newcommand{\cB}{\mathcal{B}}
\newcommand{\bbN}{\mathbb{N}}
\newcommand{\bbR}{\mathbb{R}}
\newcommand{\bbQ}{\mathbb{Q}}
\newcommand{\bbRp}{\mathbb{R}_{\geq 0}}

\newcommand{\mfF}{\mathfrak{F}}
\newcommand{\cT}{\mathcal{T}}
\newcommand{\cL}{\mathcal{L}}

\newcommand{\p}{\vec{p}}
\newcommand{\q}{\vec{q}}
\renewcommand{\r}{\vec{r}}
\newcommand{\s}{\vec{s}}
\renewcommand{\t}{\vec{t}}
\newcommand{\uu}[1]{\vec{u}^{(#1)}}   
\newcommand{\ee}[1]{\vec{e}_{#1}}     

\newcommand{\eps}{\varepsilon}
\newcommand{\eqdef}{\coloneqq}
\newcommand{\D}{\mathbb{D}}
\newcommand{\Do}{\overline{\mathbb{D}}_{\mathbb{H}}}
\newcommand{\Du}{\underline{\mathbb{D}}_{\mathbb{H}}}
\renewcommand{\H}{\mathbb{H}}
\newcommand{\majo}{\succsim}
\newcommand{\trump}{\succsim^*}
\newcommand{\cata}{\succsim_c}

\renewcommand{\vec}[1]{\mathbf{#1}}

\begin{document}

\title{Entropy and relative entropy from information-theoretic principles}

\author{Gilad~Gour and Marco~Tomamichel,~\IEEEmembership{Senior Member,~IEEE}%
\thanks{
G.~Gour is with the Department of Mathematics and Statistics, Institute for Quantum Science and Technology,
University of Calgary, Alberta, Canada.} 
\thanks{M.~Tomamichel is with the Department of Electrical and Computer Engineering as well as the Centre for Quantum Technologies, National University of Singapore, Singapore.
Email: marco.tomamichel@nus.edu.sg} 
}

\maketitle

\begin{abstract}
	We introduce an axiomatic approach to entropies and relative entropies that relies only on minimal information-theoretic axioms, namely monotonicity under mixing and data-processing as well as additivity for product distributions. We find that these axioms induce sufficient structure to establish continuity in the interior of the probability simplex and meaningful upper and lower bounds, e.g., we find that every relative entropy satisfying these axioms must lie between the R\'enyi divergences of order $0$ and $\infty$. We further show simple conditions for positive definiteness of such relative entropies and a characterisation in terms of a variant of relative trumping. Our main result is a one-to-one correspondence between entropies and relative entropies.
\end{abstract}


\section{Introduction}

There is a rich literature on axiomatic derivations of entropies and relative entropies, starting already in Shannon's seminal work~\cite{shannon48} and then refined by Faddeev~\cite{faddeev56}, Diderrich~\cite{diderrich75} and Acz\'el-Forte-Ng~\cite{aczel74}, amongst others. Such approaches first focussed on deriving the Shannon entropy until the scope was extended by R\'enyi~\cite{renyi61}. Detailed reviews of the various axiomatic derivations can be found in the books by Acz\'el-Dar\'oczy~\cite{aczel75} and Ebanks-Sahoo-Sander~\cite{ebanks98}, and a rough guide through the literature was more recently compiled by Csisz\'ar~\cite{csiszar08}. 

Some of the axioms used in the above-mentioned works can be seen as inspired by operational or information-theoretic considerations\,---\,for example the requirement that the entropy or relative entropy is additive for product distributions is not only desirable mathematically but necessary for the quantity to attain operational meaning as an information measure in an asymptotic setting where rates are considered. Some other axioms, however, lack a clear information-theoretic motivation. To see this, let us look at entropy first.

An entropy, if we want it to be compatible with our intuitive notion, should be an additive measure of uncertainty about the outcome of a random experiment. We thus expect it to be invariant under relabelling of outcomes, i.e.\ permutations of the probability distribution as well as adding and removing unused labels. In R\'enyi's derivation this invariance under permutation is required specifically. However, in both Shannon's and R\'enyi derivation of entropy, we find the following additional assumption~\cite{shannon48}:
\begin{quote}
\emph{If a choice be broken down into two successive choices, the original $\H$ should be the weighted sum of the individual values of $\H$.}
\end{quote}
This essentially fixes the rule $\H(XY) = \H(X) + \sum_{x} \Pr[X = x]\, \H(Y | X = x)$ for the joint entropy of two random variables $X$ and $Y$, and helps to single out the Shannon entropy as the unique additive uncertainty measure satisfying this rule. However, it is not evident why any meaningful additive measure of uncertainty should necessarily satisfy this. Indeed, R\'enyi went on to relax this assumption. In~\cite[Postulate 5$'$]{renyi61}, the above is replaced with a more general mean which allows for the exponential weighting of entropy contributions seen in the R\'enyi family of entropies. However, while this is useful to isolate R\'enyi entropies, it is hard to justify this axiom information-theoretically. Moreover, continuity inside the probability simplex is required explicitly by both Shannon and R\'enyi. Although this is often a very natural property for operationally meaningful information measures, we have not seen a direct operational argument for its necessity. Indeed, we will show that it follows from more directly operationally motivated axioms.

In this work we start with a different, more information-theoretically motivated set of axioms we would like entropies, divergences and relative entropies to satisfy. It is worth pointing out at this point that the nomenclature for entropies, divergences and relative entropies is not consistent throughout the literature. In the remainder of this section we present a convention that makes sense for this paper and we believe also more generally in the context of information theory and statistics. It is however at odds with how the terms are used in some of the literature. Most prominently, Tsallis entropies~\cite{tsallis88} are generally not additive and thus do not qualify as entropies in our framework. Moreover, while the terms relative entropy and divergence are often used interchangeably in the literature, we will make a distinction between them and require additivity only for relative entropies.


In the following, for an entropy function $\H(\cdot)$, which takes a probability mass function as an input, the following (see Section~\ref{sec:axioms} for a formal statement) requirements are imposed:
\begin{enumerate}
	\item it should be monotonically increasing under bistochastic (mixing) maps; and
	\item it should be additive for product distributions.
\end{enumerate}
Bistochastic maps can be interpreted as probabilistic mixtures of permutations of the outcomes due to the Birkhoff-von Neumann theorem~\cite{birkhoff46}. Forgetting which permutation was performed should not decrease the uncertainty about the outcome, and, hence, the above monotonicity property is a natural requirement for any meaningful measure of uncertainty. 
It is worth noting that this monotonicity property is often not stated as an axiom in the literature, but rather follows only once a specific expression for the entropy has been determined from the axioms.

A similar situation arises in the study of relative entropy. For a relative entropy function $\D(\cdot\|\cdot)$, which takes two probability mass functions as inputs, we deem the following two requirements essential:
\begin{enumerate}
	\item it should be monotonically decreasing under the application of a stochastic map to both arguments, i.e.\ the data-processing inequality; and
	\item it should be additive for pairs of product distributions.
\end{enumerate}
The former is necessary in most information-theoretic contexts. Let us for example consider asymmetric binary hypothesis testing where both the critical rate as well as error and strong converse exponents are characterised by relative entropies. Operationally it is evident that distinguishing outputs of a stochastic map is harder than distinguishing its inputs, and this thus needs to be reflected in any quantity that obtains operational meaning in this problem. Due to the close relation between hypothesis testing and various information-theoretic tasks (see, e.g.,~\cite{blahut74}), similar arguments can be made for many operational quantities in information theory.

The main question we ask here is how much structure these information-theoretic axioms impose on entropies and relative entropies. First, we note that since these axioms only determine entropy or relative entropy functions up to convex combinations, we cannot hope to recover a one-parameter family of functions as in the work of R\'enyi. Still, we find that the structure imposed by these axioms suffices to establish some interesting properties that all operationally meaningful entropies and relative entropies have to satisfy.

The remainder of this paper is structured as follows. After Preliminaries in Section~\ref{sec:pre} we introduce our axioms for entropies and relative entropies in Section~\ref{sec:axioms}. Section~\ref{sec:bounds} then establishes continuity and upper and lower bounds on relative entropies. Our main result then follows in Section~\ref{sec:biject}, where we prove a bijection between entropies and relative entropies (under some weak and necessary regularity assumptions). Finally, Section~\ref{sec:faithful} concerns itself with positive definiteness (or faithfulness) of relative entropies. 
In Section~\ref{sec:char} we find a characterisation of relative entropies in terms of catalytic relative majorisation.
We conclude in Section~\ref{sec:conc} by asking whether every entropy satisfying our axioms is in fact a convex combination of R\'enyi entropies.

\section{Preliminaries}
\label{sec:pre}

\subsection{Conventions and notation}

Throughout we denote by $\log$ the binary logarithm. 
We restrict our attention to finite discrete random variables on the alphabet $[n]\eqdef\{1,...,n\}$ for $n \in \bbN$. A probability mass function is represented as a row vector $\p = [p_1, p_2, \ldots, p_n]$ with $p_i \geq 0$ for all $i \in [n]$ and $\sum_{i \in [n]} p_i = 1$. We call such vectors probability vectors in the following.
The set of all such vectors is denoted by $\cP(n)$. The subset of $\cP(n)$ with strictly positive entries is denoted $\cPp(n)$. The support of a vector $\p \in \cP(n)$ is denoted  
\begin{align}
\supp(\p) \eqdef \{x\in[n]\;:\; p_x > 0\}.
\end{align}
The number of elements in the support of $\p$ is denoted by $|\p|$ and we write $\p \gg \q$ if $\supp(\p) \supseteq \supp(\q)$.
We use $\uu{n} \in \cP(n)$ to denote the uniform distribution, i.e.\ ${u_i}^{(n)} = \frac{1}{n}$ for all $i \in [n]$. On the other hand, ${\ee{i}}^{(n)}$ denotes the deterministic distribution with all mass on $i \in [n]$. We simply write $\ee{i}$ if the size of the alphabet is clear from context. Finally, $\otimes$ denotes the Kronecker (tensor) product of two vectors and $\oplus$ denotes the direct sum or concatenation of two vectors. We note that if $\p \in \cP(n)$ and $\q \in \cP(m)$ are probability vectors then $\p \otimes \q \in \cP(mn)$ is also a probability vector, whereas $\p \oplus \q$ is not.

The set of all $n \times m$ right (row) stochastic matrices, or \emph{channels}, is denoted by $\cS(n,m)$, with the shortcut $\cS(n) \eqdef \cS(n,n)$. And for any  $\p \in \cP(n)$ and $W \in \cS(n, m)$, we write $\p W \in \cP(m)$ for the output probability vector induced by the channel $W$ on input~$\p$. The set of bistochastic maps in $\cS(n)$, i.e., channels that map $\uu{n}$ to itself, is denoted by $\cB(n)$. 

We will consider functions $f: \cP(n) \to \bbRp \cap \{\infty\}$ to the extended positive real line that satisfy $f(\p) < \infty$ for any $\p \in \cPp(n)$. For such function we say that $f$ is upper semi-continuous at $x \in \cP(n)$ if, for every sequence $\{x_n\}_{n \in \bbN} \subset \cPp(n)$ that converges to $x$, we have $\limsup_{n \to \infty} f(x_n) \leq f(x)$, with the convention that $\infty \leq \infty$. We say that $f$ is lower semi-continuous at $x$ if $\liminf_{n \to \infty} f(x_n) \geq f(x)$ for all such sequences, and that $f$ is continuous at $x$ if it is  both lower and upper semi-continuous at $x$. 

\subsection{Majorisation and mixing channels}
\label{sec:majo}

For a vector $\p \in \cP(n)$ we denote by $\p^{\downarrow}$ the vector with the same components as $\p$ that are rearranged in decreasing order, i.e., the components of $\p^{\downarrow}$ satisfy $p_1^{\downarrow}\geq\cdots\geq p_n^{\downarrow}$. For convenience we also define that $p_i^{\downarrow} = 0$ for $i > n$. We say that $\p \in \cP(n)$ \emph{majorises} $\q \in \cP(m)$, and write $\p \majo \q$, if and only if $\sum_{x \in [k]} p_x^{\downarrow} \geq \sum_{x \in [k]} q_x^{\downarrow}$ for all $k \in [\max\{m,n\}]$. 
If $\p \majo \q$ and $\q \majo \p$ then we will write $\p \sim \q$.

A famous characterisation by Hardy, Littlewood and P\'olya~\cite{hardy34} states that for any two vectors $\p, \q \in \cP(n)$, we have $\p \majo \q$ if and only if there exists a bistochastic map $W \in\cB(n)$ such that $\p W = \q$. Moreover, the Birkhoff-von Neumann theorem~\cite{birkhoff46} allows us to interpret such maps as probabilistic mixtures of permutation operations, or \emph{mixing channels}. When the dimensions of $\p$ and $\q$ do not agree this can be straight-forwardly extended by allowing for maps that add symbols that have probability zero, as well as their combination with bistochastic maps.
We denote the set of $n \times m$ mixing channels by $\cB(n, m)$. Formally, we have the following equivalence, which is a simple generalisation of the result in~\cite{hardy34} that we state for completeness.

\begin{lemma} \label{lem:hardy}
	Let $\p \in \cP(n)$ and $\q \in \cP(m)$ for $n, m \in\mathbb{N}$ and let $k = |\p|$. The following statements are equivalent:
	\begin{enumerate}
		\item $\p \majo \q$;
		\item There exists a mixing channel $W \in \cB(k, m)$ acting on the support of $\p$ such that $\p W = \q$. 
	\end{enumerate}
\end{lemma}

\begin{proof}
	We first show (1)~$\implies$~(2). First, observe that $\p \majo \q$ implies $k \leq m$ and thus we can introduce a probability vector $\p' \in \cP(m)$ that is comprised of all the nonzero components of $\p$ padded with zeros. Clearly $\p' \majo \q$ and by~\cite{hardy34} there exists a bistochastic map from $\p'$ to $\q$. The reverse implication follows immediately from~\cite{hardy34} as well since adding unused symbols does not affect the majorisation condition. 
\end{proof}

\subsection{Relative majorisation}

We say that a pair of vectors $\p,\q \in \cP(n)$ \emph{relatively majorises} another pair of vectors $\p',\q' \in \cP(m)$, and write
	$(\p,\q) \majo (\p',\q')$
if there exists a channel $W \in \cS(n,m)$ such that $\p' =  \p W$ and $\q' = \q W$. For the special case where $m=n$ and $\q = \q' = \uu{n}$, this corresponds to requiring a bistochastic map $W \in \cB(n)$ such that $\p' = \p W$, and thus 
\begin{align}
(\p, \uu{n}) \majo (\p', \uu{n}) \iff \p \majo \p' \,.
\end{align}
If $(\p,\q) \majo (\p',\q')$ and $(\p',\q') \majo (\p,\q)$ then we will write $(\p,\q) \sim (\p',\q')$.

Relative majorisation is a partial order that can be characterised with testing regions. The testing region of a pair of probability vectors $\p,\q\in\cP(n)$ is a region in $\bbR^2$ defined as
\begin{align}
	\cT(\p,\q) \eqdef \Big\{ \big(\p \t^T , \q \t^T\big) \in \bbR^2 : \t \in [0,1]^n \Big\} .
\end{align}
where $\t$ is a probabilistic hypothesis test, a vector with entries between 0 and 1. This region is bounded by two curves known as lower and upper Lorenz curve. An example of a testing region is drawn in Fig.~\ref{fig:testing}. The upper Lorenz curve can be obtained from the lower Lorenz curve by a rotation of 180 degrees. Therefore, the Lower (or upper) Lorenz curve determines the testing region uniquely.

\begin{figure}[h]\centering
    \includegraphics[width=0.4\textwidth]{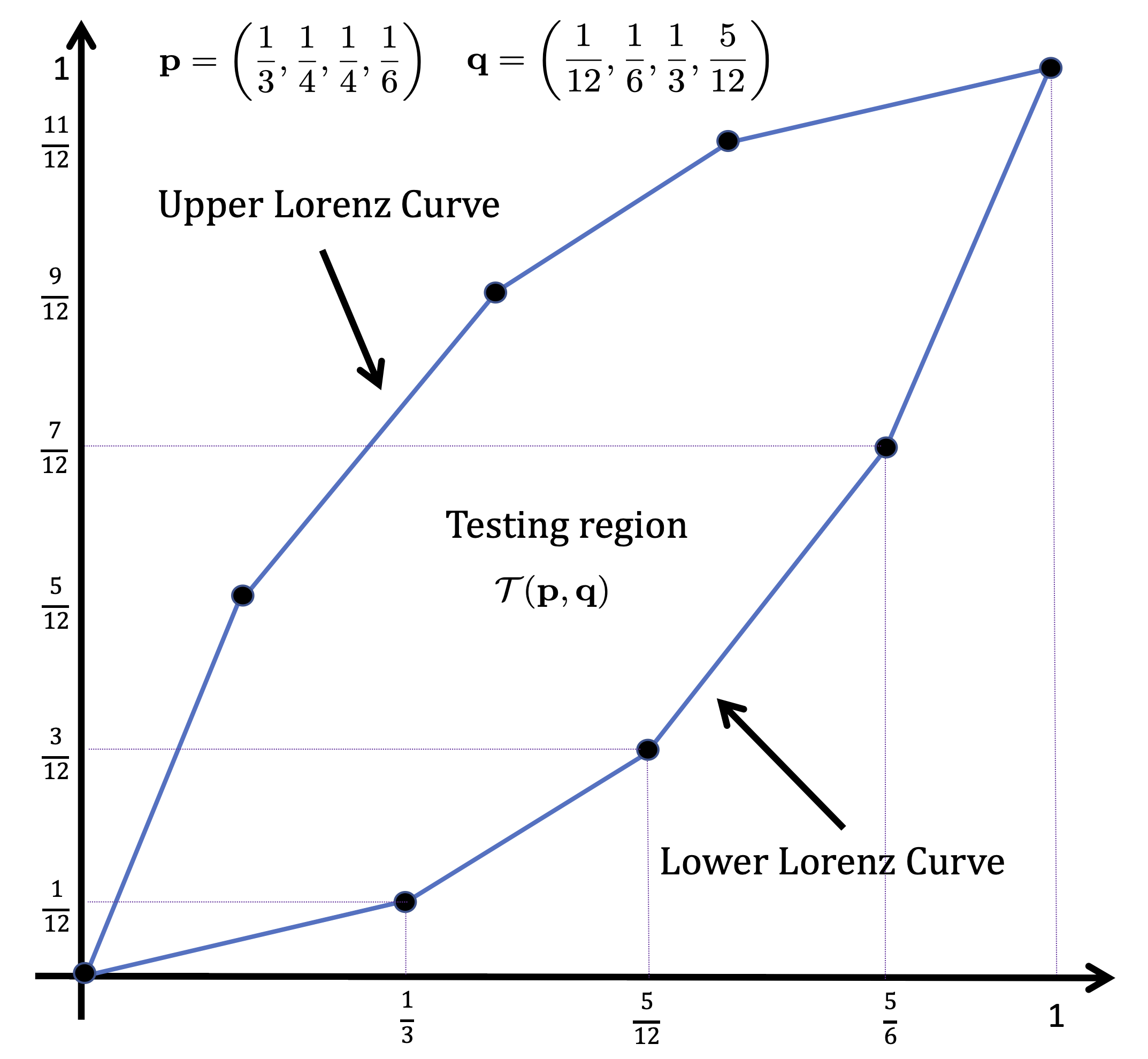}
  \caption{Let $\p, \q \in \cP(n)$ and assume for simplicity that the likelihood ratios $\lambda_i = {p_i}/{q_i}$ are ordered $\lambda_1 \geq \lambda_2 \geq \ldots \geq \lambda_n$.
  Then the lower Lorenz curve has $n+1$ vertices given by $\{(a_\ell,b_\ell)\}_{\ell=0}^{n}$
  with $a_\ell=\sum_{x=1}^{\ell} p_x$ and $b_\ell=\sum_{x=1}^{\ell} q_x$.
  The figure depicts the testing region for an example with $n=4$ and $\p$, $\q$ as given above. The vertices of the lower Lorenz curve are computed as $(0,0)$, $(1/3,1/12)$, $(7/12,3/12)$, $(5/6,7/12)$, and $(1,1)$.}
  \label{fig:testing}
\end{figure}

The relevance of testing regions to our study here is the following theorem that goes back to Blackwell~\cite{blackwell1953} and has since been rediscovered under different names including $d-$majorisation~\cite{Veinott1971}, matrix majorization~\cite{Dahl-1999}, and thermo-majorisation~\cite{HO2013} (more details can also be found in the book on majorisation by Marshall-Olkin~\cite{Marshall-2011a}).
\begin{theorem}\cite{blackwell1953}
	Let $\p,\q\in\cP(n)$ and $\p',\q'\in\cP(m)$ be two pairs of probability vectors. Then, 
	\begin{align}
		(\p,\q) \majo (\p',\q') \quad \iff \quad \cT(\p',\q') \subseteq \cT(\p,\q) \,.
	\end{align}
\end{theorem}
The theorem above provides a geometric characterisation of relative majorisation; that is,  $(\p,\q) \majo (\p',\q')$ if and only if the lower Lorenz curve of $(\p,\q)$ is nowhere above the lower Lorenz curve of $(\p',\q')$.

Finally, for any $\p \neq \q$, we note that as a consequence of the equipartition property or weak law of large numbers (see, e.g., ~\cite[Theorem~11.8.2]{cover91}) any point in $(0,1)^2$ is covered by the testing region $\cT(\p^{\otimes n}, \q^{\otimes n})$ for sufficiently large $n \in \bbN$.


\section{Axioms for entropies, divergences, and relative entropies}
\label{sec:axioms}

As will become evident later, there is a one-to-one correspondence between entropies and relative entropies. Here we however introduce their axioms independently. We will also introduce divergences, which we call quantities that satisfy data-processing but are not necessarily additive.

\subsection{Entropies} \label{sec:axiom-ent}

Here we consider a class of functions
\begin{align}
	\H: \bigcup_{n \in \bbN} \cP(n) \to \bbRp 
	\label{eq:entropy}
\end{align}
that map probability vectors in all finite dimensions to the positive reals. For entropies we have the following two main desiderata. First, as entropies are uncertainty measures, they should be non-decreasing when we apply channels that simply randomly rearrange labels. As we have seen, this relation can be captured by mixing channels and the majorisation relation.
\begin{description}
	\item[Monotonicity under mixing:] For any $n, m \in \bbN$, $\p \in \cP(n)$ and $\p' \in \cP(m)$ such that $\p \majo \p'$, we have
	\begin{align}
		\H(\p) \leq \H(\p') \,.
	\end{align}
\end{description}
Alternatively $\p'$ can be seen as the output of a mixing channel acting on the support of $\p$ (cf.\ Lemma~\ref{lem:hardy}).
Second, we require entropies to be additive for tensor products of probability distributions.
\begin{description}
	\item[Additivity:] For any $m, n \in \bbN$, $\p_1 \in \cP(n)$, and $\p_2 \in \cP(m)$, we have
	\begin{align}
		\H(\p_1 \otimes \p_2) = \H(\p_1) + \H(\p_2) \,.
	\end{align}
\end{description}

While this requirement is very natural for entropies that have an information-theoretic interpretations, it is in general not satisfied by Tsallis entropies~\cite{tsallis88}, for example.


\begin{definition}
	A function $\H$ of the form~\eqref{eq:entropy} that satisfies monotonicity under mixing and additivity, and is normalised such that  $\H(\uu{2}) = \log 2$, is called an \emph{entropy}.
\end{definition} 

A similar axiomatic definition of entropy has recently been considered in~\cite{gour19}.
The following immediate consequence of these two properties is worth pointing out.
\begin{lemma} \label{lem:entropy}
	Let $\H$ be an entropy and $n \in \bbN$. For all $i \in [n]$ and $\p \in \cP(n)$, we have 
	\begin{align}
		0 = \H \big({\ee{i}}^{(n)}\big) \leq H(\p) \leq  \H \big(\uu{n} \big) = \log n \,.
	\end{align}
\end{lemma}

\begin{proof}
	The inequalities follow from the monotonicity under mixing and the relation ${\ee{1}}^{(n)} \majo \p \majo \uu{n}$, which can easily be verified for all $i \in [n]$ and $\p \in \cP(n)$. 
	
	It remains to show the two equalities. On the one hand, using additivity we immediately find that $H(1) = 0$. Since $1 \majo {\ee{i}}^{(n)} \majo 1$ for all $n \in \bbN$, $i \in [n]$, monotonicity under mixing yields the desired equality for all deterministic distributions. On the other hand, define $f: \mathbb{N} \to \mathbb{R}$ as $f(n) = \H(\uu{n})$. By normalisation we have $f(2) = \log 2$, and, thus, by additivity,  $f(2^k) = k \log 2$ for all $k \in \bbN$. Moreover, since $\uu{n} \majo \uu{n+1}$ for all $n \in \bbN$, the function $f$ is monotonically non-decreasing. Using these properties, we can show that, for all $\ell \in \mathbb{N}$,
\begin{align}
	f(n) = \frac{1}{\ell} f( n^{\ell} )\leq \frac{1}{\ell} f\left( 2^{ \lceil \frac{\log n}{\log 2} \ell \rceil} \right) = \frac{1}{\ell} \left\lceil \frac{\log n}{\log 2} \ell \right\rceil \log 2
\end{align}
and similarly $f(n) \geq \frac{1}{\ell} \big\lfloor \frac{\log n}{\log 2} \ell \big\rfloor \log 2$. In the limit $\ell \to \infty$ both these bounds converge to $\log n$, concluding the proof.
\end{proof}


\subsection{Divergences and relative entropies}
\label{sec:divergences}

Let us now consider a class of functions 
\begin{align}
	\D: \bigcup_{n\in \bbN} \left\{\cP(n) \times \cP(n) \right\} \to \bbRp \cup \{\infty\} \label{eq:dd}
\end{align}
that map pairs of probability vectors in all finite dimensions to the positive reals or its extension to $+\infty$.  We impose two important restrictions on such functions. The first, monotonicity under data-processing, requires the divergence to be non-increasing under application of the same channel on both arguments. Intuitively we would like to think of $\D$ as a measure of distinguishability of the second argument from the first, and thus require that application of noise (modelled by a channel) cannot make the two distributions easier to distinguish.

\begin{description}
	\item[Monotonicity under data-processing:] For any $m,n \in \bbN$, $\p,\q \in \cP(n)$ and $\p', \q' \in \cP(m)$ such that $(\p,\q) \majo (\p', \q')$, we have
	\begin{align}
		\D(\p' \| \q' ) \leq \D(\p \| \q) \,. \label{eq:DPI}
	\end{align}
\end{description}
Alternatively, using the definition of relative majorisation, $\p' = \p W$ and $\q' =\q W$ can be seen as output of some channel $W$.
This is also called the \emph{data-processing inequality} (DPI). 
	
Clearly the DPI at most determines $\D$ up to additive and multiplicative constants, which we can remove by appropriate normalisation.  The multiplicative freedom usually boils down to a choice of units (e.g.\ bits or nats).
To remove the additive freedom let us start with an immediate observation about functions satisfying Eq.~\eqref{eq:DPI}. For all $m, n \in \bbN$,  $\p,\q \in \cP(n)$, and $\vec{c} \in \cP(m)$, the DPI applied for the channel with constant output $\vec{c}$ establishes that $\D(\vec{c}\|\vec{c}) \leq \D(\p\|\q)$. Hence, such functions necessarily take their minimum value whenever the two arguments agree. The natural choice for normalisation is thus to set this minimum value to be zero.
		
\begin{definition}
	A function $\D$ of the form~\eqref{eq:dd} that satisfies monotonicity under data-processing, and is normalised such that $\D(1\|1) = 0$, is called a \emph{monotone divergence}.
\end{definition}
We note that in the statistics literature the term \emph{divergence} is often used to denote faithful functionals that do not necessarily satisfy monotonicity under data-processing, the most prominent example being the Bregman divergences~\cite{bregman67}. Note that for any monotone divergence we must have $\D(\p\|\p) = 0$ for any $\p$ as discussed above. However, faithful functionals vanish if and only if the two arguments agree, a stronger requirement than what we impose here. We will discuss faithfulness in Section~\ref{sec:faithful}.
	
The second property simply requires that the functions are additive for product distributions.  
\begin{description}
	\item[Additivity:] For any $m,n \in \bbN$, $\p_1,\q_1 \in \cP(n)$, and $\p_2,\q_2 \in \cP(m)$, we have
	\begin{align}
		\D\big(\p_1 \otimes \p_2 \big\| \q_1 \otimes \q_2 \big) = \D(\p_1 \| \q_1) + \D(\p_2 \| \q_2) \,. \label{eq:additivity}
	\end{align}
\end{description}

For a relative entropy, in addition to monotonicity under data-processing and additivity, we also require normalisation. Additivity of $\D$ fixes the additive normalisation (see Lemma~\ref{lem:relent-divergence} below), so it remains to remove the multiplicative freedom. We do this by requiring that $\D(\ee{1}  \| \uu{2} ) = \log 2$. This choice is  consistent with the normalisation of R\'enyi and Kullback-Leibler divergences and, in contrast to the DPI and additivity, breaks the symmetry between the two arguments. 

\begin{definition}
	A function $\D$ of the form~\eqref{eq:dd} that satisfies both monotonicity under data-processing and additivity, and is normalised such that $\D(\ee{1}  \| \uu{2} ) = \log 2$, is called a \emph{relative entropy}.
\end{definition}

\begin{lemma}
	\label{lem:relent-divergence}
	Every relative entropy is a monotone divergence.
\end{lemma}
\begin{proof}
	Since any relative entropy satisfies the DPI, it remains to show normalisation. By additivity $\D( \p \| \q) = \D(\p \otimes 1 \| \q \otimes 1) = \D( \p \| \q) + \D(1\|1)$, and thus $\D(1\|1)$ must vanish.
\end{proof}

We can classify relative entropies depending on how they behave under an exchange of arguments, i.e.\ we say that a relative entropy is
\emph{symmetric} if $\D( \uu{2} \| \ee{1}) \in (0, \infty)$. In this case we can define its dual relative entropy, 
	\begin{align}
		\D_{*}(\p\|\q) \eqdef  \frac{\D(\q\|\p)}{\D( \uu{2} \| \ee{1})} \,.
	\end{align}
We call it \emph{asymmetric} if $\D( \uu{2} \| \ee{1}) = \infty$ and \emph{pathological} if $\D (\uu{2} \| \ee{1} ) = 0$. The latter are obviously not faithful.


\subsection{R\'enyi relative entropies and entropies}

A one-parameter family of relative entropies has been introduced by R\'enyi~\cite{renyi61}.\footnote{Note that they are usually called R\'enyi divergences in the literature, but in our framework they are called R\'enyi relative entropies.} Notably in his seminal paper R\'enyi derived the relative entopies based on a set of mathematical axioms that included additivity and equivalence under reordering, which is a special case of the data-processing inequality. However, some of the other axioms used by R\'enyi do not readily allow for an information-theoretic interpretation. R\'enyi relative entropies have found various applications in information theory, e.g., they directly characterise generalised cutoff rates in hypothesis testing~\cite{csiszar95}. 

\begin{definition}
	Let $\alpha \in (0,1) \cup (1,\infty)$. Then, for every $n \in \bbN$ and $\p, \q \in \cP(n)$, the \emph{R\'enyi relative entropy of order $\alpha$} is defined as
	\begin{align}
		D_{\alpha}(\p\|\q) \eqdef \frac{1}{\alpha-1} \log \left( \sum_{i \in [n]} p_i^{\alpha} q_i^{1-\alpha} \right) ,
	\end{align}
	whenever this expression is well-defined, and $+\infty$ otherwise.
	Moreover, the R\'enyi relative entropies of order $\alpha \in \{0,1,\infty\}$ are defined as point-wise limits.
\end{definition}

See~\cite{vanerven14} for a recent review of many more of their properties, and~\cite{vanerven10} specifically for a discussion of the relation between R\'enyi divergence and relative majorisation. All $D_{\alpha}$ with $\alpha > 0$ are continuous (in the sense introduced in the preliminaries) on $\cP(n) \times \cPp(n)$ whereas $D_0$ is trivial on $\cPp(n) \times \cPp(n)$ and has discontinuities on the boundary of the first argument.

For $n \in \bbN$ and $\p, \q \in \cP(n)$, the Kullback-Leibler relative entropy is obtained in the limit $\alpha \nearrow 1$ as
\begin{align}
  D(\p\|\q) := D_1(\p\|\q) = \sum_{x \in [n]} p_x \log \frac{p_x}{q_x} \,.
\end{align}
We in particular note that $\alpha \to D_{\alpha}(\p\|\q)$ is monotonically non-decreasing in $\alpha$. This justifies the identification
\begin{align}
 	D_{\min}(\p\|\q) &\eqdef D_{0}(\p\|\q) = -\log \sum_{i \in \supp(p) } q_i \label{eq:dmin} , \\
	D_{\max}(\p\|\q) &\eqdef D_{\infty}(\p\|\q) = \log \max_{i \in [n]} \frac{p_i}{q_i}  \label{eq:dmax} ,
\end{align}
which we call the \emph{min-relative entropy} and \emph{max-relative entropy}, respectively. One of our main results shown in the next section (see Corollary~\ref{cor:bounds}) is that these two relative entropies bound any relative entropy, not just R\'enyi relative entropies.

R\'enyi entropies can now be constructed via by the correspondence that will be discussed in detail in Section~\ref{sec:biject}. 
\begin{definition}
	Let $\alpha \in [0, \infty]$. Then, for every $n \in \bbN$ and $\p \in \cP(n)$, the \emph{R\'enyi entropy of order $\alpha$} is defined as
	\begin{align}
		H_{\alpha}(\p) := D_{\alpha}(\ee{1} \| \uu{n})  - D_{\alpha}(\p \| \uu{n}) \label{eq:renent} \,.
	\end{align}
\end{definition}
For any $\alpha \in [0,1) \cup (1,\infty)$, the expression in~\eqref{eq:renent} simplifies to the well-known formula 
\begin{align}
	H_{\alpha}(\p) = \frac{1}{1-\alpha} \log \sum_{i \in \supp(p)} p_i^{\alpha} \,.
\end{align}
The Shannon entropy~\cite{shannon48} emerges for $\alpha = 1$, and for $\alpha = \infty$ we find $H_{\infty}(\p) = -\log \min_{i \in [n]} p_i$.


\section{Continuity and bounds on relative entropies}
\label{sec:bounds}

In this section we will establish some bounds on relative entropies that will allow us to show several strong continuity properties that follow from our axioms.

\subsection{Bounds on monotone divergences}

We first establish some general bounds on monotone divergences leveraging extensively on the data-processing inequality.

\begin{definition} \label{def:minmax} 
	Let $\D$ be a monotone divergence. We define the following two derived quantities:
	\begin{align}
		\D_{\min}(\p\|\q) &\eqdef \D\big( \ee{1} \big\| [  \lambda_{\min},1 - \lambda_{\min} ] \big) , \quad  \textrm{and}\\
		\D_{\max}(\p\|\q) &\eqdef \D\big( \ee{1} \big\| [ \lambda_{\max}  ,1 - \lambda_{\max}  ] \big) ,
	\end{align}
	where $\lambda_{\min} = 2^{-D_{\min}(\p\|\q)}$ and $\lambda_{\max} = 2^{-D_{\max}(\p\|\q)}$.
\end{definition}

\begin{theorem}
	\label{thm:bounds}
	Let $\D$ be a monotone divergence. Then, $\D_{\min}$ and $ \D_{\max}$ as defined in Definition~\ref{def:minmax} are also monotone divergences.
	Furthermore, for all $n \in \bbN$ and $\p,\q \in \cP(n)$, we have
	\begin{align}
		 \D_{\min}(\p\|\q) \leq \D(\p\|\q) \leq  \D_{\max}(\p\|\q) \,. \label{eq:bounds-ineq}
	\end{align}
\end{theorem}

\begin{proof}
	To show the DPI for $\D_{\min}$ it suffices to show that $\D_{\min}(\p\|\q) \geq \D_{\min}(\p W\|\q W)$ for any channel $W \in \cS(n)$. 
	For this purpose, observe first that for any two binary probability distributions
	$\p'=(p',1-p')$ and $\q'=(q',1-q')$ there exists a channel $V \in \cS(2)$ satisfying
	\begin{align}
		\label{eq:iff2}
		\ee{1}V = \ee{1} \quad \text{and} \quad \p' V = \q'
	\end{align}
	if and only if $p' \leq q'$.\footnote{This map is trivial if $p' = 1$ and defined by the first condition in~\eqref{eq:iff2} and the relation $\ee{2}V = \alpha \ee{1} + (1-\alpha) \ee{2}$ with $\alpha = \frac{q'-p'}{1-p'}$ otherwise.}
	Moreover, since $D_{\min}(\p\|\q)$ is a monotone divergence, the DPI ensures that
	\begin{align}
		2^{-D_{\min}(\p\|\q)} &\leq 2^{-D_{\min}(\p W\|\q W)} \,.
	\end{align}
	Hence, by~\eqref{eq:iff2} there must exist a channel $V$ that keeps $\ee{1}$ intact and
	satisfies
	\begin{align}
		&\big[ 2^{-D_{\min}(\p\|\q)}, 1 - 2^{-D_{\min}(\p\|\q)} \big] V \notag\\
		&\quad = \big [2^{-D_{\min}(\p W\|\q W)}, 1 - 2^{-D_{\min}(\p W\|\q W)}\big] \,.
	\end{align}
	A close inspection of the respective definitions of $\D_{\min}(\p\|\q)$ and $\D_{\min}(\p W\|\q W)$ then reveals that the desired relation
	follows from the DPI of $\D$ applied for the channel $V$. 
	We assert that the proof for $\D_{\max}$ follows analogously. 

	We next show the two inequalities in~\eqref{eq:bounds-ineq}. To determine the lower bound on $\D(\p\|\q)$, we define the channel $E \in \cS(n,2)$ via its action on $\vec{r} \in \cP(n)$ as
	\begin{align}
		\vec{r} E = \underbrace{ \sum_{x \in \supp(\p)} r_x }_{2^{-D_{\min}(\p\|\vec{r})}} \ee{1} + 
		 \underbrace{\sum_{x\not\in\supp(\p)} r_x}_{1-2^{-D_{\min}(\p\|\vec{r})}}  \ee{2}  \,.
	\end{align}
	In particular, $\vec{p} E = \ee{1}$. Hence, the DPI for $\D$ reveals that
	\begin{align}
		\D(\p\|\q) \geq \D\big(E\p\| E\q\big) = \D_{\min}(\p\|\q) \,.
	\end{align}
	
	For the upper bound on $\D(\p\|\q)$, we recall the shorthand $\lambda_{\max} = 2^{-D_{\max}(\p\|\q)}$ from Definition~\ref{def:minmax}, and note that $\lambda_{\max} \in (0,1]$ and $\q \geq \lambda_{\max} \p$ element-wise by definition of $D_{\max}$. Consider first the case $\lambda_{\max} < 1$.
	Define now a channel $F \in \cS(2,n)$ whose rows are given by 
	\begin{align}
	\ee{1} F = \p, \quad \textrm{and} \quad \ee{2} F = \frac{\q - \lambda_{\max} \p}{1 - \lambda_{\max}} \,.
	\end{align}
	Defining now $\tilde{\q} \in \cP(2)$ as $\tilde{\q} =  [\lambda_{\max}, 1 - \lambda_{\max}]$,
	we then observe that $\tilde{\q}F = \q$. Hence, the DPI of $\D$ for $F$ yields
	\begin{align}
		\D_{\max}(\p\|\q) = \D(\ee{1}\|\tilde{\q}) \geq \D(\p\|\q) \,.
	\end{align}
	If $\lambda_{\max} = 1$, we can deduce that $\p = \q$ and thus all monotone divergences vanish, concluding the proof.
\end{proof}

\subsection{Bounds on relative entropies}

For relative entropies we can simplify the expressions for $\D_{\min}$ and $\D_{\max}$ further. For this purpose we next establish a general expression for relative entropies when the first argument is deterministic.

\begin{lemma} \label{lem:deterministic}
Let $\D$ be a relative entropy. Then, for any probability vector $\p\in\cP(n)$ we have
\begin{align}
	\D \left (\ee{x} \middle\| \p \right) = -\log p_x \quad \forall x \in [n] \,.
\end{align}
\end{lemma}

\begin{proof}
It is sufficient to show that $\D(\ee{1} \| \p) = -\log p_1$. Let $S=[\ee{1}, \s, \cdots, \s]^T \in \cS(n)$ with $\s \in \cP(n)$ arbitrary. This implies that $\ee{1}S=\ee{1}$ and we define $\q_{\s}(p_1) \eqdef \p S =  p_1 \ee{1} + (1-p_1) \s$. In particular, the choice $\s' = \frac{1}{1-p_1} [0, p_2, \ldots, p_n]$ yields $\q_{\s'}(p_1) = \p$. Applying the DPI twice, first with any $\s$ such that $s_1 = 0$ and then with $\s'$, we find
\begin{align}
	\D(\ee{1}\|\p) \geq \D(\ee{1} \| p_1 \ee{1} + (1-p_1) \s ) \geq \D(\ee{1}\|\p) ,
\end{align}
and thus equality holds. 

This means that $\D(\ee{1}\|\p)$ is independent of $p_2, \ldots, p_n$, and we may define $f: [0,1] \to \bbRp \cup \{\infty\}$ by $f(p_1) = \D(\ee{1}\|\p)$. The function has the following two properties. 
\begin{enumerate}
\item
Since for any channel $T$ satisfying $\ee{1} T =\ee{1}$ we have $\p T = p_1\ee{1}+(1-p_1)\t$ for some $\t \in \cP(n)$, we can conclude that the first component of $\p T$ cannot be smaller than $p_1$. The DPI thus ensures that $f$ is monotonically non-increasing. 
\item
Additivity of $\D$ implies that $f$ itself is additive, i.e., for any $x,y \in [0,1]$, we have $f(xy)=f(x)+f(y)$.
\end{enumerate}

Define now $g: n \mapsto f\left(\frac{1}{n}\right)$ as function on natural numbers $n \in \bbN$, which is non-decreasing and additive. Therefore, due to Erd\"os theorem, $g(n)=c \log(n)$ for some constant $c \in \mathbb{R}$. The normalisation condition for relative entropies reads $g(2)=1$, and thus $c=1$. Moreover, for any integer $m\leq n$, additivity implies that
\begin{align}
\log(m)+f\Big(\frac{m}{n}\Big) 
&= f \Big(\frac{1}{m} \Big) + f\Big(\frac{m}{n}\Big) \\
&= f \Big(\frac{1}{n} \Big) = \log(n) ,
\end{align}
and, thus, $f\big(\frac{m}{n}\big)=\log(n)-\log(m)=-\log\big(\frac{m}{n}\big)$.  Hence, the function is determined for all rational numbers in $[0,1]$.
Finally, for any $r\in[0,1]$ let $\{q_k\}, \{p_k\}$ be two sequences of rational numbers in $[0,1]$ with limit $r$ and $q_k<r<p_k$ for all $k \in \bbN$. Such sequences always exists since the rational numbers are dense in $\bbR$. Now, the monotonicity of $f$ yields
\begin{align}
	-\log q_k = f(q_k) \geq f(r) \geq f(p_k) = -\log p_k \;.
\end{align}
Taking the limit $k\to\infty$ on both sides and using the continuity of $\log$ we get $f(r)=-\log(r)$, concluding the proof.
\end{proof}

The following is therefore an immediate consequence of Theorem~\ref{thm:bounds} and Lemma~\ref{lem:deterministic}.
\begin{corollary}
	\label{cor:bounds}
	Let $\D$ be a relative entropy. Then, for all $n \in \bbN$ and $\p,\q \in \cP(n)$, we have
	\begin{align}
		D_{\min}(\p\|\q) \leq \D(\p\|\q) \leq D_{\max}(\p\|\q) \,.
	\end{align}
\end{corollary}
\noindent In particular, these bounds imply that
	\begin{itemize}
	\item  $\D(\p \| \q) < \infty$ if $\supp(\p) \subseteq \supp(\q)$, and 
	\item $\D(\p\|\q) = \infty$ if $\supp(\p) \cap \supp(\q) = \emptyset$,
	\end{itemize}
inheriting these properties from $D_{\max}$ and $D_{\min}$, respectively.

\subsection{Continuity of relative entropy}
\label{sec:cont}

Ideally we would like relative entropies to be continuous functions of probability vectors, but this is not always ensured. Most prominently, $D_{\min}$ exhibits jumps at the boundary. We are however able to show several strong continuity properties that follow from our axioms.

Our main tool is a triangle inequality for relative entropies.
\begin{theorem} \label{thm:triangle}
	Let $\D$ be a relative entropy. For all $n \in \bbN$ and $\p, \q, \t \in \cP(n)$, we have
	\begin{align}
		\D(\p \| \q) \leq \D(\p \| \t) + D_{\max}(\t \| \q) \,.
	\end{align}
\end{theorem}
Note that this reduces to the upper bound in Corollary~\ref{cor:bounds} when we set $\t = \p$.

\begin{proof}
	We may write, for an appropriate choice of $\eps \in (0, 1]$, 
	\begin{align}
		\q = (1-\eps) \t + \eps \r, \quad \textrm{where} \quad \r = \t + \frac{1}{\eps} (\q - \t) \,.
	\end{align}
	Note that to assert that $\r \in \cP(n)$ we need to ensure that $\r \geq 0$ and, thus, $\q \geq (1-\eps) \t$ entry-wise. This holds if $\eps = 1 - 2^{-D_{\max}(\t\|\q)}$, by definition of the max-relative entropy. Using additivity of $\D$ and Lemma~\ref{lem:deterministic}, followed by the DPI for $\D$, we find that
	\begin{align}
		&\D(\p\|\t) - \log ( 1 - \eps ) =  \D \big( \p \otimes [1, 0] \big\| \t \otimes [1-\eps, \eps] \big) \\
		&\qquad \geq \D \big( (\p \otimes [1, 0]) W \big\| (\t \otimes [1-\eps, \eps]) W \big) \,.  \label{eq:cont1}
	\end{align}
	Here the DPI is applied for a channel $W$ that acts as an identity upon detecting $[1, 0]$ in the second register, and produces a constant output $\r$ upon detecting $[0,1]$ in the second register, i.e.\ the channel $W$ is defined by
	\begin{align}
		(\ee{i} \otimes [1,0] )W = \ee{i} , \quad (\ee{i} \otimes [0,1]) W = \r, \quad \forall i \in [n] \,.
	\end{align}
	Clearly then $(\p \otimes [1, 0]) W = \p$ and $(\t \otimes [1-\eps, \eps]) W = \q$, and hence Eq.~\eqref{eq:cont1} establishes that
	\begin{align}
		\D(\p \| \q) - \D(\p \| \t )  \leq - \log (1-\eps) = D_{\max}(\t \| \q) \,,
	\end{align}
	concluding the proof.
\end{proof}

This can be used to show that all relative entropies are continuous in the interior of $\cP(n) \times \cP(n)$. 

\begin{corollary} \label{cor:upperlowercont}
	Let $\D$ be a relative entropy and $n \in \bbN$. Then, $(\p, \q) \mapsto \D(\p\|\q)$ is upper semi-continuous on $\cP(n) \times \cPp(n)$ and continuous on $\cPp(n) \times \cPp(n)$.
\end{corollary}

\begin{proof}
	Consider sequences $\{ \p_k \}_{k \in \bbN}, \{ \q_k \}_{k \in \bbN} \subset \cPp(n)$. We first show upper semi-continuity. Due to DPI and Theorem~\ref{thm:triangle}, we have 
	\begin{align}
		\D(\p\|\q) &\geq \D(\p_k \| \q W_k) \\
		&\geq \D(\p_k \| \q_k) - D_{\max}(\q W_k \| \q_k), \label{eq:cont3}
	\end{align} 
	where $W_k \in \cS(n)$ is a channel given by
	\begin{align}
		\ee{i} W_k = (1-\eps_k) \ee{i} + \eps_k \p + (\p_k - \p) , \quad \forall i \in [n], \label{eq:cont4}
	\end{align}
	where $\eps_k = 1 - 2^{-D_{\max}(\p\|\p_k)}$ so that, similar to the proof of Theorem~\ref{thm:triangle}, $W_k$ indeed describes a stochastic map. Clearly $\p W_k = \p_k$. Since $\p \ll \p_k$ for all $k \in \bbN$ we get $\lim_{k \to \infty} \eps_k = 0$. And thus $\lim_{k \to \infty} \q W_k = \q$. Hence, using~\eqref{eq:cont3}, we find
	\begin{align}
		\limsup_{k \to \infty} \D(\p_k \| \q_k) &\leq \D(\p\|\q) + \limsup_{k \to \infty} D_{\max}(\q W_k \| \q_k) \,,
	\end{align}
	and the latter limit vanishes due to the continuity of $D_{\max}$ at the point $(\q, \q)$ which by assumption is in the interior of $\cP(n) \times \cP(n)$.
	
	For lower semi-continuity, we use the bounds
	\begin{align}
		\D(\p_k\|\q_k) &\geq \D(\p \| \q_k W_k) \\
		&\geq \D(\p\|\q) - D_{\max}(\q_k W_k \| \q_k)
	\end{align}
	where $W_k$ is given analogously to Eq.~\eqref{eq:cont4} but with the roles of $\p$ and $\p_k$ interchanged. Note in particular that we need $\p \gg \p_k$ for $W_k$ to be well-defined, which is given by our assumption that $\p$ has full support. By taking $\liminf_{k \to \infty}$ on both sides we show lower semi-continuity.
\end{proof}

Critical behaviour at the boundary of $\cP(n) \times \cP(n)$ when the supports are not identical is expected as some relative entropies experience jumps there. On the one hand, the relative entropy $D_{\textrm{path}}(\p\|\q) \eqdef D_{\min}(\p\|\q) + D_{\min}(\q\|\p)$ is not lower semi-continuous at such points. On the other hand, $D_{\max}(\p\|\q)$ is not upper-semicontinuous when $\q$ does not have full support, which can be seen by considering the limit of the sequences $\{ \p_k \}_{k \in \bbN}$, $\{ \q_k \}_{k \in \bbN}$ with $\p_k = \big(\frac1k, 1 - \frac1k \big)$, $\q_k = \big( \frac1{k^2}, 1 - \frac1{k^2} \big)$.

We can give more specific bounds, for example in terms of the Schatten $\infty$-norm distance, which is given as $\| \p - \q \|_{\infty} \eqdef \max_{i \in [n]} | p_x - q_x |$ for $\p, \q \in \cP(n)$. Also recall that $\p^{\downarrow}_n$ and $\q^{\downarrow}_n$ denote the smallest entries of $\p$ and $\q$, respectively.

\begin{corollary} \label{cor:cont}
	Let $\D$ be a relative entropy, $n \in \bbN$ and $\p \in \cP(n)$. Then, for $\q, \tilde{\q} \in \cPp(n)$ with $ \| \q - \tilde{\q} \|_{\infty} <  \min \{ q^{\downarrow}_n, \tilde{q}^{\downarrow}_n \}$, we have
	\begin{align}
		\left| \D(\p\|\q) - \D(\p\|\tilde{\q}) \right| \leq \log \left(1 + \frac{ \| \q - \tilde{\q} \|_{\infty} }{ \min \{ q^{\downarrow}_n, \tilde{q}^{\downarrow}_n \} } \right) \,. \label{eq:cont2}
	\end{align}
\end{corollary}

In particular, the function $\q \to \D(\p\|\q)$ is continuous on $\cPp(n)$ for all $\p \in \cP(n)$ and not only strictly positive $\p$, strengthening the result of Corollary~\ref{cor:upperlowercont} when we only consider the relative entropy as a function of the second argument.

\begin{proof}
	To verify the inequality is suffices to note that $2^{D_{\max}(\tilde{\q} \| \q)} \leq 1 +   (q^{\downarrow}_n)^{-1} \| \tilde{\q} - \q \|_{\infty}$. Eq.~\eqref{eq:cont2} then follows by symmetry and ensures continuity on $\cPp(n)$. 
\end{proof}


\section{Bijection between entropies and relative entropies}
\label{sec:biject}

Here we show a strict one-to-one correspondence between continuous entropies and continuous relative entropies.

\begin{theorem} \label{thm:bijection}
	There exists a bijection $\mfF$ with inverse $\mfF^{-1}$ mapping between relative entropies that are continuous in the second argument and entropies.
\end{theorem}

The form of the bijection is given in Propositions~\ref{prop:relent-to-ent} and~\ref{prop:ent-to-relent}.
The proof is split into two parts. First we show how to construct an entropy from a relative entropy. This construction is mostly standard but we repeat it here as our definition of monotonicity under mixing operations, required for entropies, is a bit more restrictive than monotonicity under bistochastic maps that is usually considered in the literature.
\begin{prop}
	\label{prop:relent-to-ent}
Given a relative entropy $\D$, define $\mfF(\D)$ of the form~\eqref{eq:entropy} as follows. For all $n \in \bbN$ and $\p \in \cP(n)$, define
		\begin{align}
			\mfF(\D): \p \mapsto &\ \D\big( \ee{1} \| \uu{n} \big) - \D \big( \p \| \uu{n} \big) \label{eq:ent-from-relent} \\
			= &\ \log n - \D \big( \p \| \uu{n} \big) \,.
		\end{align}
		Then, $\mfF(\D)$ is an entropy.
\end{prop}

\begin{proof}
	We need to show that $\H = \mfF(\D)$ as defined in Eq.~\eqref{eq:ent-from-relent} is an entropy. First note that additivity and normalisation immediately follow from the respective properties of the relative entropy~$\D$. It thus remains to show monotonicity under mixing. As discussed in Section~\ref{sec:majo}, majorisation $\p \majo \q$ between $\p \in \cP(n)$ and $\q \in \cP(m)$ holds if and only if there exists a bistochastic map $W \in \cB(\max\{m,n\})$ that takes $\p$ to $\q$, where the vectors are padded with $0$'s as required. As such, it suffices to show 
\begin{enumerate}
	\item Monotonicity under bistochastic maps: $\H(\p) \leq \H( \p W)$ for any bistochastic map $W$.
	\item Equality under embedding: $\H(\p \oplus 0) = \H(\p)$.
\end{enumerate}
The first item is an immediate consequence of the DPI of the relative entropy since $\uu{n} W = \uu{n}$ for any bistochastic map. Equality under embedding can be verified as follows. First, from Eq.~\eqref{eq:ent-from-relent} we see that $H\big( \ee{1}^{(n)} \big) = 0$ for any $n \in \mathbb{N}$. Hence, since we have already established additivity, we find that for any $\p \in \cP(n)$ it must hold that
\begin{align}
	\H(\p) = \H\Big(\p \otimes \ee{1}^{(n+1)} \Big) &= \H\Big( \big( \p \oplus 0 \big) \otimes \ee{1}^{(n)} \Big) \\
	&= \H(\p \oplus 0) \,.
\end{align}
	This concludes the proof.
\end{proof}

Next, we construct a relative entropy from an entropy. 
This construction is more involved and we start by introducing two candidate extensions.

\begin{definition}
\label{def7}
For any $n \in \bbN$ and $\p, \q \in \cP(n)$ and any entropy function~$\H$, we define
the \emph{maximal extension of $\H$} and the \emph{minimal extension of $\H$}, respectively, as
\begin{align}
\Do(\p\|\q) &\eqdef& \inf :\quad &  \log k - \H(\r) \label{eq:Do} \\
		 &&\mathrm{subject~to}:\quad & \big(\r, \uu{k}\big) \majo (\p, \q)  \notag\\
		&&& k \in \bbN \;,\;  \r \in \cP(k) , \quad \textrm{and} \notag\\
\Du(\p\|\q) &\eqdef & \sup :\quad & \log k - \H(\r) \label{eq:Du} \\
		 &&\mathrm{subject~to}:\quad & (\p, \q) \majo \big(\r, \uu{k}\big)  \notag\\
		&&& k\in \bbN \;,\; \r\in\cP(k) \,. \notag
\end{align}
\end{definition}
We start by showing some elementary properties of these two quantities. The expressions maximal and minimal extension of $\H$ are justified by Property (2) below. 

\begin{lemma} \label{lem:prop}
Let $\Do$ and $\Du$ be the two extensions of an entropy function $\H$ as defined in Definition~\ref{def7}. Then,
\begin{enumerate}
	\item[(1)] Both $\Do$ and $\Du$ are monotone divergences, and for all $n \in \mathbb{N}$ and $\p, \q \in \cP(n)$, we have
	$\Do(\p\|\q) \geq \Du(\p\|\q)$. Moreover, for all $m \in \bbN$ and $\p', \q' \in \cP(m)$, we have
	\begin{align}
		\Du(\p \otimes \p' \|\q \otimes \p') &\geq \Du(\p\|\q) + \Du(\p'\|\q') \,, \label{eq:add1} \\
		\Do(\p \otimes \p'\|\q \otimes \p') &\leq \Do(\p\|\q) + \Do(\p'\|\q') \label{eq:add2} \,.
	\end{align}

	\item[(2)] For any $n \in \bbN$ and $\p\in\cP(n)$, we have 
	\begin{align}
		\label{eq:constraint}
		\Do(\p\|\uu{n}) = \Du(\p\|\uu{n}) = \log n - \H(\p) .
	\end{align}
	Moreover, for any monotone divergence $\D$ satisfying $\D(\p\|\uu{n}) = \log n - \H(\p)$ for all $\p \in \cP(n)$, we have
	\begin{align} \label{sim}
		\Du(\p\|\q) \leq \D(\p\|\q) \leq \Do(\p\|\q) \quad \forall\;\p,\q \in \cP(n)\;.
	\end{align}
	\end{enumerate}
\end{lemma}

\begin{proof}[Proof of Lemma~\ref{lem:prop}, Property (1)]
	It is easy to verify that $\Do(1\|1) = \Du(1\|1) = 0$ (see also Property 2) for a proof). The data-processing inequality $\Do(\p \| \q) \geq \Do(\p W\|\q W)$ for any $W \in \cB(n,m)$ follows from the implication
	\begin{align}
		\big(\r, \uu{k}\big) \majo (\p, \q)  \implies \big(\r, \uu{k}\big) \majo (\p W, \q W) \,,
	\end{align}
	which allows us to relax the constraint in Eq.~\eqref{eq:Do} to find a lower bound on $\Do(\p \| \q)$.
	Similarly, the DPI for $\Du$ follows from the implication
	\begin{align}
		(\p W, \q W) \majo \big(\r, \uu{k}\big) \implies (\p, \q) \majo \big(\r, \uu{k}\big) \,.
	\end{align}
	
	Next, to see that $\Do(\p\|\q) \geq \Du(\p\|\q)$, we observe that for any candidates $(k, \r)$ and $(k', \r')$ in the optimisation in Eqs.~\eqref{eq:Do} and~\eqref{eq:Du}, respectively, we must have $(\r, \uu{k}) \majo (\r', \uu{k'})$. Thus, by tensoring the input and output vectors by $\uu{k'}$ and $\uu{k}$ we find $(\r \otimes \uu{k'}, \uu{kk'}) \majo (\r' \otimes \uu{k}, \uu{kk'})$, which is equivalent to the majorisation relation
	\begin{align}
		\r \otimes \uu{k'} \majo \r' \otimes \uu{k} \,. \label{eq:majo1}
	\end{align} 
	Now further note that the inequality we want to establish,
	\begin{align}
		\log k - \H(\r) \geq \log k' - \H(\r') \,,
	\end{align}
	is equivalent to $\H(\r \otimes \uu{k'}) \leq \H(\r' \otimes \uu{k})$ using the additivity of $\H$ and Lemma~\ref{lem:entropy}. But this holds due to monotonicity under mixing of $\H$ and Eq.~\eqref{eq:majo1}, concluding this part of the proof.

	Finally, the two inequalities follow immediately by restricting the optimisations in Eqs.~\eqref{eq:Do} and~\eqref{eq:Du} to product vectors and channels, and leveraging the additivity of $\H$.
\end{proof}

\begin{proof}[Proof of Lemma~\ref{lem:prop}, Property (2)] This follows from the observation that if $\q=\uu{n}$ then in both optimisations in Eqs.~\eqref{eq:Do} and~\eqref{eq:Du} the choice $k=n$, $\r=\p$, and $E$ the identity channel satisfies the constraints. Hence, using Property 1), we get
\begin{align}
	\log n - \H(\p) &\geq \Do(\p \| \uu{n}) \\
	&\geq \Du(\p \| \uu{n}) \geq \log n - \H(\p) \,.
\end{align}

	It remains to show that any monotone divergence $\D$ with $\D(\p\|\uu{n}) = \log n - \H(\p)$ is sandwiched in between $\Do$ and $\Du$.
	On the one hand, since $\D$ satisfies the DPI, we get, for any $k \in \mathbb{N}$ and $F \in \cS(n, k)$ such that $\q F = \uu{k}$.
	\begin{align}
		\D(\p\|\q) \geq \D(\p F \| \q F) = \log k - \H(\p F) \,.
	\end{align}
	Maximising over all such maps we find $\D(\p\|\q) \geq \Du(\p\|\q)$. On the other hand, we note that
	\begin{align}
		\Do(\p\|\q) = \inf \left\{ \D(\r\|\uu{k}) : \r E = p , \uu{n} E = q \right\}
	\end{align}
	where $k \in \bbN$, $\r \in \cP(k)$ an $E \in \cS(k, n)$. The desired inequality follows since $\D(\r\|\uu{k}) \geq \D( \r E \| \uu{n} E) = \D(\p\|\q)$.
\end{proof}

Next we will need a very useful technical lemma.
\begin{lemma}
\label{lem:embed}
Let $n \in \bbN$, $\p \in \cP(n)$ and $\q \in \cPp(n) \cap \bbQ^n$. Then, there exists $k \in \bbN$ and $\r \in \cP(k)$ such that
\begin{align}
(\p,\q) \sim (\r,\uu{k}) \,. \label{eq:sim}
\end{align}
The vector $\r$ can be written as $\r = \bigoplus_{x=1}^n p_x \uu{k_x}$ where
$\q$ is assumed to be of the form
\begin{align}
	\q = \left( \frac{k_1}{k}, \frac{k_2}{k}, \ldots, \frac{k_n}{k} \right) \,.
\end{align}
for some $\{k_x\}_{x \in [n]} \subset \bbN$ with $\sum_{x=1}^{n}k_x=k$.
\end{lemma}

\begin{proof}
First note that relative majorisation is invariant under permutation of indices, and thus we can without loss of generality assume that
\begin{align}
	\frac{p_1}{q_1}\geq\frac{p_2}{q_2}\geq \ldots \geq \frac{p_n}{q_n} \,.
\end{align}
Let us next recall that $(\p,\q) \majo (\p',\q')$ if and only if the lower Lorenz curve of $(\p,\q)$ is no where above the lower Lorenz curve of $(\p',\q')$. The lower Lorenz curve of $(\p,\q)$ is constructed as prescribed in Fig.~\ref{fig:testing} and denoted $\cL(\p,\q)$. 
The curve is comprised of affine segments connecting $n+1$ vertices given by tuples $\{(a_\ell,b_\ell)\}_{\ell=0}^{n}$, where 
\begin{align}
	a_\ell = \sum_{x=1}^{\ell} p_x \quad \text{and} \quad b_\ell=\sum_{x=1}^{\ell} q_x \,.
\end{align}
Our assumption that $\q \in \cPp(n)$ ensures that the $\cL(\p,\q)$ does not contain any horizontal segments.

Now, choose $k$ large enough such that the set $\{\frac{i}{k}\}_{i \in [k]}$ contains all the rational numbers $\{b_{\ell}\}_{\ell=1}^{m}$. This is possible since $\q \in \bbQ^n$ and, consequently, $\{b_{\ell}\}$ are rational as well. We next define $\r \in \cP(k)$ as follows. For any $i \in \{0, 1, \ldots, k\}$ define $s_i$ to be the $x$-axis coordinate that corresponds to the $y$-axis coordinate $\frac{i}{k}$ of $\cL(\p,\q)$. That is, $s_i$ is the unique number satisfying $(s_i, \frac{i}{n}) \in \cL(\p,\q)$. See also Figure~\ref{LorenzCurve} for an example. We then define $r_i \eqdef s_{i} - s_{i-1}$ for all $i \in [k]$; 

\begin{figure}[h]\centering
    \includegraphics[width=0.4\textwidth]{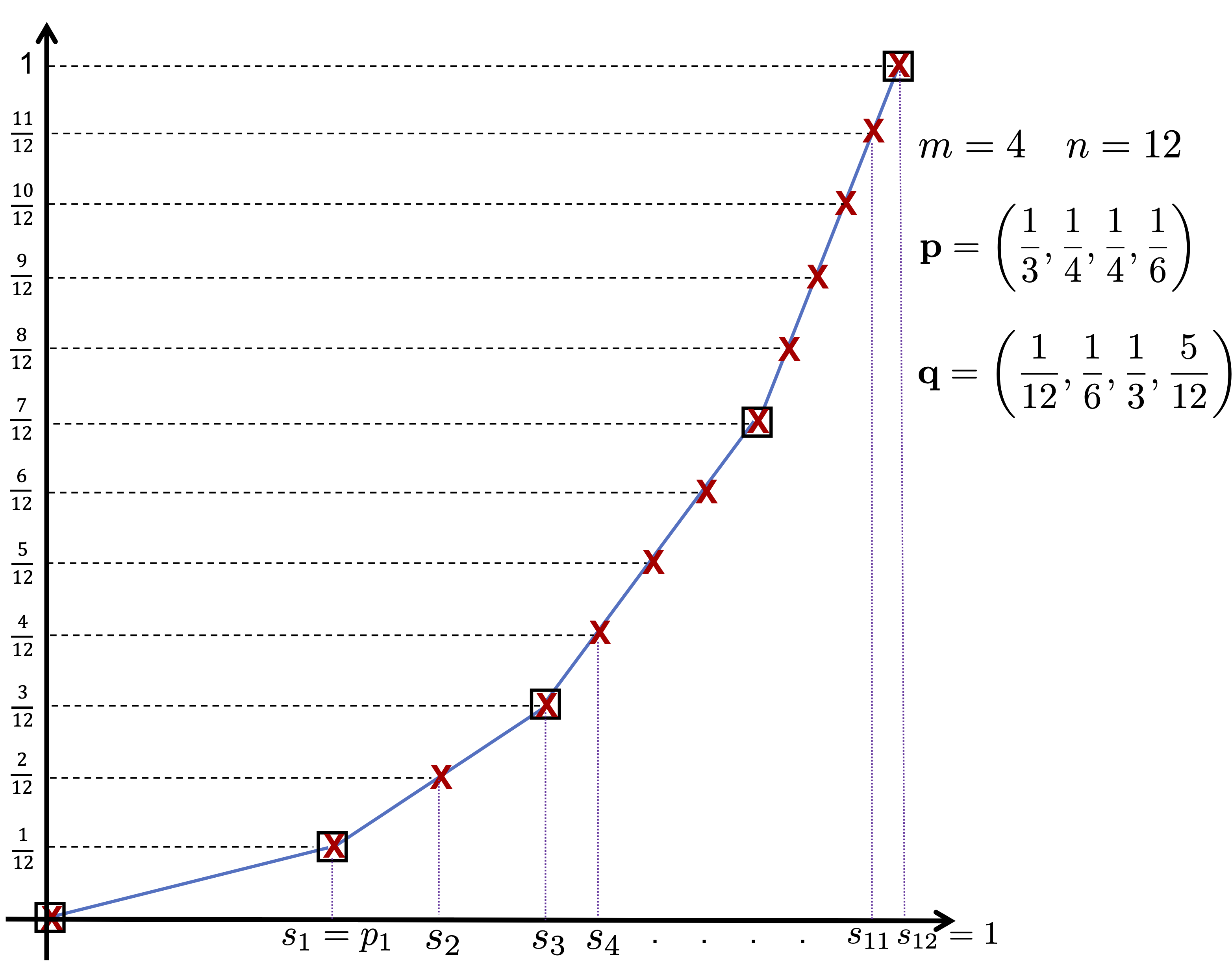}
  \caption{The lower Lorenz curve of $(\p,\q)$ and $(\r,\uu{k})$. In this example, the black squares are the vertices of $\cL(\p,\q)$ (total of 5) and the red x's (total of 13) are the vertices of $(\r,\uu{k})$.}
  \label{LorenzCurve}
\end{figure}
By construction, all the vertices of $\cL(\r,\uu{k})$ are also in $\cL(\p,\q)$, and since  $b_\ell \in \{\frac{i}{k}\}_{i=1}^{n}$ it follows that we also have $(a_\ell,b_\ell) \in \cL(\r,\uu{k})$. Hence, we conclude that $\cL(\p,\q)=\cL(\r,\uu{k})$ which means that for this choice of $k$ and $\r$ we have $(\p,\q) \majo (\r,\uu{k}) \majo (\p,\q)$ as required.
\end{proof}

The following is a direct consequence of Lemma~\ref{lem:embed} and has first been proposed in~\cite[Lemma 13]{wehner13} for R\'enyi relative entopies, but we include a different proof here for general monotone divergences that we hope is also illuminating.

\begin{lemma}
	Let $\D$ be a monotone divergence, $n \in \bbN$ and $\q \in \cP_+(n) \cap \bbQ^n$. Then there exits $k \in \bbN$ such that the following holds. For every $\p \in \cP(n)$ there exists an $\r \in \cP(k)$ such that
	\begin{align}
		\D(\p \| \q) = \D \big( \r \| \uu{k} \big)  \,.
	\end{align}
\end{lemma}

\begin{proof}
	Since $\q \in \bbQ^n$, we may write $q_x = \frac{k_x}{k}$ with $k_x \in \bbN$ for all $x \in [n]$, where $k \in \bbN$ is any (or the smallest) common denominator of the rationals $\{q_x\}_{x\in [n]}$. Now consider the dilution channel $W \in \cS(n, k)$ given as follows. First, let $s_{\ell} = \sum_{i=1}^{\ell} k_i$ and set $s_0 = 0$. Then, we define
	\begin{align}
		W(j|i) = \begin{cases}
			 	\frac{1}{k_i}	& \textrm{when} \quad j \in \{s_{i-1}, s_{i-1} + 1, \ldots, s_i\} , \\
			 0 & \textrm{elsewhere}  \,.
			 \end{cases}
	\end{align}
	This channel simply dilutes the input symbol to ensure that $\q W = \uu{k}$. We now simply take $\r = \p W$. Since $W$ is invertible the DPI ensures the desired equality.
\end{proof}

The following map from entropies to relative entropies that are continuous in the second argument is a consequence of Lemmas~\ref{lem:prop} and~\ref{lem:embed}.

\begin{prop}
	\label{prop:ent-to-relent}
	Given an entropy $\H$, define $\mfF^{-1}(\H)$ of the form~\eqref{eq:dd} as follows. For all $n \in \bbN$, $\p \in \cP(n)$ and $\q \in \cPp(n) \cap \bbQ^n$, define
		\begin{align}
			\mfF^{-1}(\H)(\p, \q) := \log k - \H(\r), \label{eq:relent1}
		\end{align}
		where $k$ and $\r$ are constructed as in Lemma~\ref{lem:embed} such that $(\p,\q) \sim (\r,\uu{k})$.
		For general $\q \in \cP(n)$, $\mfF^{-1}(\H)$ is defined via continuous extension.
		
		Then, $\mfF^{-1}(\H)(\p\|\q) = \Du(\p\|\q) = \Do(\p\|\q)$ when $\q \in \cPp(n) \cap \bbQ^n$ is rational.
		Moreover, $\mfF^{-1}(\H)$ is a relative entropy and continuous in $\q$ for any fixed $\p$.
\end{prop}

\begin{proof}
	Let $\D = \mfF^{-1}(\H)$. We first verify the identities $\D(\p\|\q) = \Du(\p\|\q) = \Do(\p\|\q)$ for $n \in \bbN$, $\p \in \cP(n)$ and $\q \in \cPp(n) \cap \bbQ^n$. Using Lemma~\ref{lem:embed} we can deduce that $(\r, \uu{k})$ is a feasible solution for the optimisation in the definition of $\Du(\p\|\q)$ and $\Do(\p\|\q)$, and therefore
	\begin{align}
		\Do(\p\|\q) \leq \log k - \H(\r) \leq \Du(\p\|\q) \,.
	\end{align}
	But since $\Du(\p\|\q) \leq \Do(\q\|\q)$ equality must hold.
	
	The quantity we defined is thus a monotone divergence and satisfies the normalisation condition for relative entropies since $\H(\ee{1}) = 0$. It furthermore satisfies both Eq.~\eqref{eq:add1} and~\eqref{eq:add2}, which ensure that it is additive on the restricted space where $\q$ is rational and has strictly positive entries. 
	
	It remains to show that $\q \to \D(\p\|\q)$ is continuous on $\cPp(n) \cap \bbQ^n$, so that its continuous extension to $\cP(n)$ is well-defined. To verify this note that the argument leading to Theorem~\ref{thm:triangle}, and specifically Corollary~\ref{cor:cont}, remains valid if we restrict the second argument to be rational.
	Finally, note that data-processing inequality and additivity are preserved under continuous extension and thus the resulting quantity is indeed a relative entropy, concluding the proof.
\end{proof}

Propositions~\ref{prop:relent-to-ent} and~\ref{prop:ent-to-relent} finally establish Theorem~\ref{thm:bijection}. It is worth pointing out that $D_{\textrm{path}}(\p\|\q) = D_{\min}(\p\|\q) + D_{\min}(\q\|\p)$, which we already encountered as a counterexample to lower semi-continuity in Section~\ref{sec:cont}, is mapped to $H_{\max}$ by $\mfF$ and $H_{\max}$ is in turn mapped to $D_{\min}(\p\|\q)$ by its inverse $\mfF^{-1}$; hence we are losing the contribution $D_{\min}(\q\|\p)$ that is discontinuous in~$\q$ in the process. This explains why the requirement that relative entropies be continuous in the second argument for the bijection is crucial.

Finally, we observe that the correspondence between relative entropies and entropies allows to port certain results from relative entropies to entropies.

\begin{corollary}
	Let $\H$ be an entropy and $n \in \bbN$. Then, $\H$ is continuous on $\cPp(n)$ and lower semi-continuous everywhere. Moreover, for all $\p \in \cP(n)$, we have
	\begin{align}
		H_{\min}(\p) \leq \H(\p) \leq H_{\max}(\p),
	\end{align}
	where  the min-entropy is given by $H_{\min}(p) \eqdef -\log p^{\downarrow}_n$ and the max-entropy is the Hartley entropy $H_{\max}(\p) \eqdef \log |\p|$.
\end{corollary}

We have now seen how relative entropies can be constructed from entropies. In Appendix~\ref{app:schurconvex} we show that monotone divergences can be constructed in a similar way from continuous Schur-convex functions.


\section{Faithfulness}
\label{sec:faithful}

We have already noted in Section~\ref{sec:divergences} that every normalised monotone divergence (and thus every relative entropy) $\D$ satisfies $\D(\p\|\q) = 0$ if $\p = \q$.
Faithfulness of a relative entropy refers to the property that this equality holds if and only if $\p = \q$. Not all relative entropies are faithful. For example, $D_{\min}(\p\|\q) = 0$ for any $\p$ with $\supp(\p) \supseteq \supp(\q)$. However, as we show now, $D_{\min}$ is a very unique relative entropy, and almost all other relative entropies are faithful. 

Before we characterise faithful relative entropies, we define the following order parameter for relative entropies. 
\begin{definition} 
	\label{def:order}
	Let $\D$ be a relative entropy. We first introduce the function $f_{\D}(\eps) \eqdef \D(\vec{u} + \eps \vec{\Delta} \| \vec{u})$ with the vectors $\vec{\Delta} = \big(\frac12, -\frac12\big)$ and $\vec{u} = \uu{2}$. Then $f_{\D}(\eps)$ is well-defined for $|\eps| \leq 1$. Next, we define the \emph{order} $\alpha_{\D}$ of $\D$ as
	\begin{align}
		\alpha_{\D} \eqdef \liminf_{\eps \to 0} \left\{ \frac{ f_{\D}(\eps) - 2 f_{\D}(0) + f_{\D}(-\eps)}{ \eps^2} \right\}  \label{eq:order}
	\end{align}
\end{definition}
Note that this is simply the lower second derivative of $f_{\D}$ at $\vec{u}$ in direction $\vec{\Delta}$. Moreover, the expression in Eq.~\eqref{eq:order}
can be simplified using symmetry under permutation, $f_{\D}(\eps) = f_{\D}(-\eps)$, and the fact that $f_{\D}(0) = 0$, to get
\begin{align}
	\alpha_{\D} =  2 \liminf_{\eps \to 0} \left\{ \frac{ f_{\D}(\eps) }{\eps^2} \right\} \,.  \label{eq:alphasimp}
\end{align}
By Lemma~\ref{lm:renyi2nd} in Appendix~\ref{app:renyi}, we can conclude that for R\'enyi relative entropies we have $\alpha_{D_{\alpha}} = \alpha$ for all $\alpha \in [0, \infty]$, as intended.

We give the following characterisation of relative entropies that are not faithful.
\begin{theorem} \label{thm:faithful}
  Let $\D$ be a relative entropy. The following three statements are equivalent: 
  \begin{itemize}
  	\item[(a)] $\D$ is not faithful;
	\item[(b)] $\D(\p\|\q) = 0$ for all $n \in \bbN$ and $\p, \q \in \cP(n)$ with $\supp\{\p\} = \supp\{\q\}$;
 	\item[(c)] $\alpha_{\D} = 0$.
  \end{itemize}
  Moreover, the above statements also imply:
  \begin{itemize}
	\item[(d)] $\p \mapsto \D(\p\|\q)$ is not lower semi-continuous in $\p$.
  \end{itemize}
\end{theorem}
Therefore, if any of the statements (b), (c) or (d) are false for a relative entropy $\D$, then $\D$ is faithful. In particular, all lower semi-continuous relative entropies are faithful.

\begin{proof}[Proof of  Theorem~\ref{thm:faithful}, (a)$\implies$(b)]
	Since $\D$ is not faithful there must exist two distinct probability distributions $\tilde{\vec{p}} \neq \tilde{\vec{q}}$ such that $\D(\tilde{\vec{p}} \|\tilde{\vec{q}}) = 0$. 
	
	We first show the desired statement for binary distributions. Since $\tilde{\vec{p}}^{\times n}$ and $\tilde{\vec{q}}^{\times n}$ are asymptotically perfectly distinguishable, for any $\eps, \delta \in (0,1)$ we can find a suitable $n \in \mathbb{N}$ and probabilistic hypothesis test $T: \{0,1\}^{\times n} \to [0,1]$ such that $\Pr_{\tilde{\vec{p}}^{\otimes n}}[T] = 1-\eps$ and $\Pr_{\tilde{\vec{q}}^{\otimes n}}[T] = \delta$. Hence, first additivity and then data-processing of $\D$ reveal that
	\begin{align}
	  0 &= \D(\tilde{\vec{p}} \| \tilde{\vec{q}}) = \D \big(\tilde{\vec{p}}^{\otimes n} \big\| \tilde{\vec{q}}^{\otimes n} \big) 
	  \geq \D \big( [1-\eps, \eps] \big\| [\delta, 1-\delta] \big) ,
	\end{align}
	and, thus, the quantity on the right vanishes. Since $\eps, \delta \in (0,1)$ are arbitrary, this is what we aimed to show.  
	
	We now proceed to show the statement for general $\vec{p}$, $\vec{q}$ with equal support by contradiction. Assume $\D(\vec{p}\|\vec{q}) > 0$. We can always choose $0<s<t<1$ such that
	\begin{align}
		\frac{1-t}{1-s}\vec{p} \leq\vec{q} \leq \frac{t}{s} \vec{p} \,. \label{eq:defts}
	\end{align}
	holds entry-wise by taking $s$ close enough to $0$ and $t$ close enough to $1$. Now define the stochastic channel 
	\begin{align}
		E := \begin{bmatrix} \frac{1}{t-s}\left((1-s)\vec{q}-(1-t)\vec{p} \right),\ \frac{1}{t-s}\left(t\vec{p}-s\vec{q}\right) \end{bmatrix},
	\end{align}
	and note that the conditions in Eq.~\eqref{eq:defts} ensure that the matrix is positive (entry-wise). Moreover, for $\vec{s}=(s,1-s)$ and $\vec{t}=(t,1-t)$ we have by definition $E \vec{s} = \vec{p}$ and $E\vec{t} = \vec{q}$. Hence, using the DPI of $\D$ we get $\D(\vec{s}\|\vec{t}) \geq \D(\vec{p}\|\vec{q}) > 0$, contradicting the statement shown in the previous paragraph.
\end{proof}

\begin{proof}[Proof of  Theorem~\ref{thm:faithful}, (b)$\implies$(c) $\land$ (d)]
	Statement~(b) ensures that $f_{\D}(\eps) = 0$ for all $\eps \in [0, 1)$. Thus,
	$\alpha_{\D} = 0$. Moreover, since normalisation requires $f_{\D}(1) = 1$, lower semi-continuity is violated.
\end{proof}

\begin{proof}[Proof of  Theorem~\ref{thm:faithful}, (c)$\implies$(a)]
	We show the contrapositive by finding a lower bound on $f_{\D}(\eps)/\eps^2$ that is independent of $\eps$ by leveraging faithfulness, additivity and the DPI of $\D$. Set $\vec{p} = \vec{u} + \eps \vec{\Delta}$ as in the definition of $f_{\D}$ above. Let us further introduce the channels $\{E_n\}_n$, for $n \in \mathbb{N}$,  which output a binary distribution with the sums of the $2^{n-1}$ smallest and largest probabilities of $p^{\otimes n}$, respectively. Clearly, we have $\vec{u}^{\otimes n} E_n = \vec{u}$ and, for odd $n$, $\vec{p}^{\otimes n} E_n = [ t_n, 1-t_n ]$ with
	\begin{align}
		t_n :=& \sum_{k=0}^{\frac{n-1}{2}}{n\choose k}\Big(\frac12 + \eps\Big)^k \Big(\frac12 - \eps \Big)^{n-k} \leq \exp(-2n\eps^2) \,,
	\end{align}
	where the inequality follows immediately from the Hoeffding bound. We now choose an odd integer $n \in \big[\frac{1}{2\eps^2}, \frac{1}{\eps^2}\big]$ 
	so that $t_n \leq \exp(-1) < \frac12$ and introduce $\delta = \frac12 - \exp(-1)$. Using additivity and the DPI of $\D$, we can now establish that
	\begin{align}
		f_{\D}(\eps) 
		&= \frac{1}{n} \D \left( \vec{p}^{\otimes n} \middle\| \vec{u}^{\otimes n} \right) 
		\geq \frac{1}{n} \D\left( \begin{bmatrix}  t_n \\ 1- t_n \end{bmatrix}  \middle\| \vec{u} \right) \\
		&= \frac1{n} f_{\D}\Big(\frac12 - t_n\Big) \,.
	\end{align}
	We then note that $f_{\D}$ is monotonically non-decreasing due to the DPI of $\D$. Hence, we can further bound the above as
	\begin{align}
		f_{\D}(\eps) \geq  \frac{1}{n} f_{\D}(\delta) \geq \eps^2 f_{\D}(\delta).
	\end{align}
	It remains to note that $f_{\D}(\delta) > 0$ since we assumed faithfulness of $\D$. Hence, using the expression for $\alpha_{\D}$
	in Eq.~\eqref{eq:alphasimp}, we establish that $\alpha_{\D} \geq f_{\D}(\delta) > 0$.
\end{proof}



\section{Relative trumping and relative entropies}
\label{sec:trumping}\label{sec:char}

We first extend the notion of trumping to the setting of pairs of probability distributions and then show that a more robust pre-order is more operationally meaningful. We then show that this pre-order, which we call catalytic relative majorisation, holds if and only if all relative entropies are ordered.

\subsection{Relative trumping}

By definition all relative entropies are monotone under relative majorisation, which follows directly from the data-processing inequality; however, we can find much weaker relations under which relative entropies are still monotone. 

First, we consider probability distributions. We say that $\p$ \emph{trumps} $\q$, and write $\p \trump \q$ if there exists a vector (known as catalyst) $\r \in \cP(\ell)$ of some finite dimension $\ell \in \bbN$ such that $\p\otimes\r \majo \q\otimes\r$. Clearly majorisation implies trumping but the converse is not true in general. We can now also introduce a trumping relation between pairs of probability distributions.

\begin{definition}
	Let $n, m \in \bbN$, $\p, \q \in \cP(n)$ and $\p', \q' \in \cP(m)$. We say that a tuple $(\p, \q)$ trumps a tuple $(\p', \q')$, and write $(\p, \q) \trump (\p', \q')$, if there exist $k \in \bbN$ and $\r, \t \in \cP(k)$ where $t$ has full support such that
	$(\p \otimes \r, \q \otimes \t) \majo (\p' \otimes \r, \q' \otimes \t)$.
\end{definition}

Evidently for every relative entropy $(\p, \q) \trump (\p', \q')$ implies $\D(\p\|\q) \geq \D(\p'\|\q')$, which follows from additivity and the data-processing inequality. However, in contrast to relative majorisation this definition is not necessarily robust as we discuss in the next subsection.

\subsection{Robustness under small perturbations}

The relative majorisation relation is robust under small perturbation. More precisely, let $\{\p_k\}_{k\in \bbN}$, $\{\q_k\}_{k\in \bbN}$ , $\{\p_k'\}_{k\in \bbN}$  and $\{\q_k'\}_{k\in \bbN}$ be sequences of probability vectors in $\cP(n)$ with limits $\p$, $\q$, $\p'$ and $\q'$, respectively. If $(\p_k, \q_k) \majo (\p_k', \q_k')$ for all $k\in\bbN$ then necessarily we have $(\p, \q) \majo (\p', \q')$. To see this we can for example invoke the characterisation in terms of testing regions and note that the inclusion relation between these regions is robust when taking the limit.

However, for the relative trumping relation this robustness property does not necessarily hold. The reason stems from the dimension of the catalyst, which can increase with $k$. Without invoking additional arguments, one cannot conclude that there exists a finite dimensional probability vectors $\r$ and $\s$ with the property that $(\p \otimes \r, \q \otimes \s) \majo (\p' \otimes \r, \q' \otimes \s)$.
This motivates us to replace the relative trumping relation with a more operationally motivated pre-order that is robust to small perturbations. We call it catalytic relative majorisation. 

\begin{definition} \label{def:cata}
Let $m,n\in\bbN$, $\p, \q \in\cP(m)$, and $\p', \q' \in\cP(n)$. We say that $(\p, \q)$ \emph{catalytically majorizes} $(\p', \q')$, and write $(\p, \q) \cata (\p', \q')$, if there exist sequences 
\begin{align}
	\{\p_k\}_{k},  \{\q_k\}_{k}  \subset \cP(m) \quad \textrm{and} \quad
	\{\p_k'\}_{k}, \{\q_k'\}_{k} \subset \cP(n), 
\end{align}
such that $(\p_k, \q_k) \trump (\p_k', \q_k')$ for all $k\in\bbN$  and in the limit $k \to \infty$ we have
$\p_k \to \p$, $\q_k \to \q$, $\p_k' \to \p'$ and $\q_k' \to \q'$.
\end{definition}

Note that the relation $\cata$ is indeed a pre-order and if $(\p, \q) \trump (\p', \q')$ then necessarily $(\p, \q) \cata (\p', \q')$ while the converse is not necessarily true. Further, the relation $\cata$ is robust under small perturbations essentially by definition.

\begin{lemma} \label{lem:perturb}
Let $m, n \in \bbN$, $\p, \q \in \cP(m)$, and $\p', \q' \in \cP(n)$. Further, let $\{\p_k\}_{k \in \bbN}$ and $\{ \q_k \}_{k \in \bbN}$ be two sequences satisfying $\p_k \to \p$ and $\q_k \to \q$,  as $k \to \infty$. Then, 
\begin{align}
	 (\p_k, \q_k) \cata (\p', \q')\ \forall k \in \bbN &\implies  (\p,\q)\cata(\p',\q') , \\
	 (\p',\q')\cata(\p_k,\q_k)\ \forall k \in \bbN &\implies  (\p',\q')\cata(\p,\q) .
\end{align}
\end{lemma}


Another question we might ask is whether it is necessary to consider four different sequences of states in the definition of catalytic relative majorisation. In Appendix~\ref{app:restrict} we show that under certain support conditions it suffices to only consider two sequences converging to $\q$ and $\q'$, respectively, while the other probability vectors can be kept fixed.

\subsection{Characterisation of relative entropies}

We now want to establish a characterisation of relative entropies in terms of catalytic relative majorisation. Remarkably, catalytic relative majorisation can be fully characterised in terms of relative R\'enyi entropies. 

\begin{theorem} \label{thm:char}
  	Let $m, n \in \bbN$, $\p, \q \in \cPp(n)$, and $\p', \q' \in \cPp(m)$. Then the following statements are equivalent:
	\begin{itemize}
		\item[(a)] $(\p, \q) \cata (\p', \q')$
		\item[(b)] $\D(\p\|\q) \geq \D(\p'\|\q')$ and $\D(\q\|\p) \geq \D(\q'\|\p')$ for every relative entropy $\D$.
		\item[(c)] $D_{\alpha}(\p\|\q) \geq D_{\alpha}(\p'\|\q')$ and $D_{\alpha}(\q\|\p) \geq D_{\alpha}(\q'\|\p')$ for every $\alpha \geq \frac12$.
	\end{itemize}
\end{theorem}

An equivalence similar to (a)~$\iff$~(c) was first claimed in~\cite{wehner13}, although for a slightly different definition of catalytic majorisation where only $\p'$ is a limit of a sequence and $\p$, $\q$ and $\q'$ are fixed. We were unable to close a gap in their proof and thus provide a full derivation here.

The implication (c)~$\implies$~(b) is new and tells us that if all R\'enyi relative entropies are ordered then in fact all relative entropies are ordered.

\begin{proof}[Proof of Theorem~\ref{thm:char}, (a)~$\implies$~(b)]
	Catalytic relative majorization implies (cf.\ Definition~\ref{def:cata}) that there exist sequences $\{\p_k\}_{k}$,  $\{\q_k\}_{k}$, $\{\p_k'\}_{k}$ and $\{\q_k'\}_{k}$, with limits $\p$, $\q$, $\p'$ and $\q'$ such that $(\p_k, \q_k) \trump (\p_k', \q_k')$ for all $k \in \bbN$. For every relative entropy~$\D$, due to additivity and the DPI, the trumping relation thus implies
	\begin{align}
		\D(\p_k \| \q_k) \geq \D(\p_k' \| \q_k') \quad \textrm{and} \quad \D(\q_k \| \p_k) \geq \D(\q_k' \| \p_k') \,.
	\end{align}
	Taking the limit $k \to \infty$ together with our continuity result in Corollary~\ref{cor:upperlowercont} then yields the desired statement.
\end{proof}

The implication (b)~$\implies$~(c) is trivial, and it thus remains to show (c)~$\implies$~(a). We will instead show a slightly stronger theorem that has weaker assumptions on the support of the probability vectors.

\begin{theorem}\label{thm:char2}
  	Let $m, n \in \bbN$, $\p, \q \in \cP(n)$, and $\p', \q' \in \cP(m)$ be such that either $\p$ or $\q$ have full support. Then the following statements are equivalent:
\begin{enumerate}
		\item[(a)] $(\p, \q) \cata (\p', \q')$
		\item[(c)] $D_{\alpha}(\p\|\q) \geq D_{\alpha}(\p'\|\q')$ and $D_{\alpha}(\q\|\p) \geq D_{\alpha}(\q'\|\p')$ for every $\alpha \geq \frac12$.\end{enumerate}
\end{theorem}

\begin{proof}[Proof of Theorem~\ref{thm:char2}, (a)~$\implies$~(c)]
Let $\p_k,\q_k,\p_k',\q_k'$ as in Definition~\ref{def:cata}. 
Hence, $(\p_k,\q_k)\trump(\p_k',\q_k')$ so that 
\begin{align}
&D_{\alpha}(\p_k\|\q_k)\geq D_{\alpha}(\p_k'\|\q_k')\quad\text{and}\label{1g}\\
&D_{\alpha}(\q_k\|\p_k)\geq D_{\alpha}(\q_k'\|\p_k')\label{2g}\;.
\end{align}

Consider first the case $\alpha\in(0,1)$. Then, $D_\alpha$ is continuous in $\cP(n)\times\cP(n)$, so that~\eqref{1g} and~\eqref{2g} imply in the limit $k\to\infty$ that $D_{\alpha}(\p\|\q)\geq D_{\alpha}(\p'\|\q')$ and $D_{\alpha}(\q\|\p)\geq D_{\alpha}(\q'\|\p')$. We therefore consider now the case $\alpha\geq 1$.
Since we assume that $\q>0$ for sufficiently large $k$ also $\q_k>0$ and therefore the limit $\lim_{k\to\infty}D_{\alpha}(\p_k\|\q_k)$ exists and equals to $D_\alpha(\p\|\q)<\infty$ for all $\alpha\geq 1$. Therefore, taking the liminf on both sides of~\eqref{1g}
gives
\begin{align}
D_{\alpha}(\p\|\q)\geq \liminf_{k\to\infty}D_{\alpha}(\p_k'\|\q_k')\geq D_{\alpha}(\p'\|\q')
\end{align}
where the second inequality follows from the lower semi-continuity of $D_\alpha$ (see, e.g.,~\cite{vanerven14}).
It is left to show $D_{\alpha}(\q\|\p) \geq D_{\alpha}(\q'\|\p')$. Observe that since $\q>0$, if $\p\not>0$ then $D_{\alpha}(\q\|\p)=\infty\geq D_{\alpha}(\q'\|\p')$. On the other hand, if $\p>0$ then taking the liminf on both sides of~\eqref{2g} gives
\begin{align}
D_{\alpha}(\q\|\p)&=\liminf_{k\to\infty}D_{\alpha}(\q_k\|\p_k)\nonumber\\
&\geq \liminf_{k\to\infty}D_{\alpha}(\q_k'\|\p_k')\nonumber\\
&\geq D_{\alpha}(\q'\|\p')
\end{align}
where the first equality follows from the fact that $D_\alpha$ is continuous on $\cP_{>0}(n)\times\cP_{>0}(n)$, and the last inequality follows from the lower semi-continuity of $D_\alpha$.
\end{proof}

To show the other direction we will need the following lemma, which is a reformulation of the Turgut-Klimesh characterisation of the trumping relation~\cite{turgut07,klimesh07}. We present it in a symmetric form that we find instructive and that has not yet appeared in the literature.

\begin{lemma}[cf.~\cite{turgut07,klimesh07}]
\label{lem:turgut}
Let $n\in\bbN$ and $\p,\p'\in\cP(n)$ with $\p \neq \p'$ and either $\p$ or $\p'$ have full support. Then the following statements are equivalent:
\begin{enumerate}
\item[(1)] $\p  \trump \p'$.
\item[(2)] For every $\alpha \geq \frac12$, we have
\begin{align}
	&D_{\alpha} \big(\p\big\|\uu{n}\big) > D_{\alpha}\big(\p'\big\|\uu{n}\big)\label{a1}
	\quad \textrm{and} \\
	 &D_{\alpha}\big(\uu{n}\big\|\p\big) >  D_{\alpha}\big(\uu{n}\big\|\p'\big) \,.
\end{align}
\end{enumerate}
\end{lemma}

The range of $\alpha$ in (2) can be extended to all $\alpha > 0$ by simply noting that $D_{\alpha}(\p \| \uu{n}) = \frac{\alpha}{1-\alpha} D_{1-\alpha}(\uu{n}\|\p)$.
Further note that the domain and support restrictions do not restrict the applicability of the theorem. For any $k, m \in \bbN$, $\tilde{\p} \in \cP(k)$ and $\tilde{\p}' \in \cP(m)$ we can construct $\p, \q \in \cP(n)$ with $n = \max\{|\tilde{\p}|, |\tilde{\p}'|\}$ in the form expected by Lemma~\ref{lem:turgut} by adding and removing zeros so that $\tilde{\p} \trump \tilde{\p}' \iff \p \trump \p'$. 

\begin{proof}[Proof of Theorem~\ref{thm:char2}, (c)~$\implies$~(a)]
In the following we assume $\p \neq \q$ and $\p' \neq \q'$ as the implication is trivial otherwise.
Due to the symmetry in the roles of $\p$ and $\q$, we can assume w.l.o.g.\ that $\q>0$. Further, as shown in Appendix~\ref{app:lemmas}, there exist sequences $\{ \q_k \}_{k \in \bbN}$ and $\{ \q_k' \}_{k \in \bbN}$ with $\q_k\in \cPp(n) \cap \bbQ^n$ and $\q_k'\in \cPp(m) \cap \bbQ^m$ such that  $\q_k \to \q$ and $\q_k' \to \q'$ as $k \to \infty$, and for all $k\in\bbN$,
\begin{align}
(\p,\q_k)\majo(\p,\q)\quad\text{and}\quad (\p',\q')\majo(\p',\q_k')\,.
\end{align}
Therefore, we have
\begin{align}
D_{\alpha}(\p\|\q_k)\geq D_{\alpha}(\p\|\q)\geq D_{\alpha}(\p'\|\q')\geq D_{\alpha}(\p'\|\q_k') ,
\end{align}
where the second inequality is the assumption in (c).
Similarly
\begin{align}
D_{\alpha}(\q_k\|\p)\geq D_{\alpha}(\q\|\p)\geq D_{\alpha}(\q'\|\p')\geq D_{\alpha}(\q_k'\|\p')
\end{align}
Now, since both $\q_k$ and $\q_k'$ have positive rational components, there exists two finite dimensional probability vectors $\r_k,\r_k'\in\cP(m_k)$ with $m_k\in\bbN$ (cf.~Lemma~\ref{lem:embed}), such that 
\begin{align}
&\big(\r_k,\uu{m_k}\big)\sim(\p,\q_k)\nonumber\\
&\big(\r_k',\uu{m_k}\big)\sim(\p',\q_k')\;.\label{conpq}
\end{align}
Hence, for all $\alpha\geq \frac12$
\begin{align}
&D_{\alpha}\big(\r_k\big\|\uu{m_k}\big) \geq D_{\alpha}\big(\r_k' \big\|\uu{m_k}\big)\quad\text{and}\nonumber\\
& D_{\alpha}\big(\uu{m_k}\big\|\r_k\big) \geq D_{\alpha}\big(\uu{m_k} \big\|\r_k'\big)\;.
\end{align}

Let us now first consider the case $\r_k' = \uu{m_k}$. Consulting the construction in Lemma~\ref{lem:embed}, we see that $\r_k' = \uu{m_k}$ implies $\p' = \q_k'$. However, this cannot occur for sufficiently large $k$ since $\q_k' \to \q' \neq \p'$ by our assumption.
Hence, we can assume $\r_k' \neq \uu{m_k}$. Note then that for any $k\in\bbN$, we have 
\begin{align}
\r_k'\majo \s_k\eqdef\left(1-\frac1k\right)\r_k'+\frac1k\uu{m_k} \,.
\end{align}
Then, by appealing to the implication (1)~$\implies$~(2) in Lemma~\ref{lem:turgut} and noting that $\s_k$ has full support, we get the following strict inequalities for all $\alpha\geq \frac12$:
\begin{align}
&D_{\alpha}\big(\r_k\big\|\uu{m_k}\big) > D_{\alpha}\big(\s_k \big\|\uu{m_k}\big)\quad\text{and}\nonumber\\
& D_{\alpha}\big(\uu{m_k}\big\|\r_k\big)> D_{\alpha}\big(\uu{m_k} \big\|\s_k\big)\,
\end{align}

Since the condition above is equivalent to the condition given in Lemma~\ref{lem:turgut} it follows that $\r_k\trump \s_k$. 
Hence, 
\begin{align}\label{99}
(\p,\q_k)\sim(\r_k,\uu{m_k})\trump (\s_k,\uu{m_k})\sim(\p'_k,\q_k')\;,
\end{align}
where 
\begin{align}
\p_k'\eqdef\left(1-\frac1{k}\right)\p'+\frac1{k}\q_k'\;.
\end{align}
The equivalence $(\s_k,\uu{m_k})\sim(\p'_k,\q_k')$ can be verified from the construction in Lemma~\ref{lem:embed}.
Since catalytic majorisation is robust to small perturbations (cf.~Lemma~\ref{lem:perturb}), taking the limit $k\to \infty$ in Eq.~\eqref{99} gives $(\p,\q)\cata(\p',\q')$, concluding the proof.
\end{proof}


\section{Conclusion}
\label{sec:conc}

We have shown that a rich framework for entropies and relative entropies can be derived from only a few axioms, including monotonicity under data-processing (relative entropy) and mixing (entropy) and additivity for product distributions. These axioms are information-theoretically meaningful and even necessary for many applications of entropies and relative entropies in information theory. Our approach thus stands in contrast to most other work on axiomatic derivations of entropies where more mathematical axioms have been taken as a starting point. 

We leave open what we believe to be a very interesting and nontrivial question, namely whether our axioms for entropies restrict us to convex combinations of R\'enyi entropies. All the properties we have shown are consistent with this hypothesis. It is worth pointing out two such properties in particular. First, in Section~\ref{sec:bounds} we have shown that all relative entropies are bounded from below and above by the minimal R\'enyi relative entropy and maximal R\'enyi relative entropy, respectively, and similarly for R\'enyi entropies. 
Second, in Section~\ref{sec:faithful} we were able to show faithfulness for all relative entropies with order parameter strictly larger than $0$, exactly as we would expect for convex combinations of R\'enyi divergences. A similar axiomatic derivation of $\ell_p$ norms via their multiplicative property, related to the additivity of R\'enyi entropies, has recently been achieved in~\cite{aubrun11}. Their techniques however do not seem to readily apply here, in particular because our axioms only restrict entropies up to convex combinations.

If the conjecture is true, it would also imply that all relative entropies that are continuous in the second argument are convex combinations of R\'enyi relative entropies.

\paragraph*{Acknowledgements} GG acknowledges support form the Natural Sciences and Engineering Research Council of Canada (NSERC). MT is supported by NUS startup grants (R-263-000-E32-133 and R-263-000-E32-731) and by the National Research Foundation, Prime Minister's Office, Singapore and the Ministry of Education, Singapore under the Research Centres of Excellence programme.





\appendices

\section{Divergences from Schur-convex functions}
\label{app:schurconvex}

In Section~\ref{sec:biject} we presented a map that constructs relative entropies from entropies. Here we develop a similar construction for divergences. 
For this purpose, let $g: \bigcup_{n \in \bbN} \cP(n) \to \bbR \cup \{\infty\}$ be a function with the following properties:
\begin{enumerate}
	\item For every $n \in \bbN$, the function $\cP(n) \ni \p \to g(\p)$ is \emph{Schur convex}\footnote{A function $g$ is Schur convex if $\p \majo \q$ implies $g(\p) \geq g(\q)$ for all $\p, \q \in \cP(n)$ and all $n \in \mathbb{N}$.} and continuous on $\cP(n)$. 
	\item For $n = 1$ it is normalised to $g(1) = 0$.
	\item For all $k, n \in \bbN$ and $\r \in \cP(n)$,
	\begin{align}
		g\big( \r \otimes \uu{k} \big) = g(\r) \,.
	\end{align}
\end{enumerate}

A large class of such functions can be constructed as follows, and we will see that they correspond to Csisz\'ar's $f$-di\-ver\-gences~\cite{csiszar72}. Given a convex function $f: \bbRp \to \bbR \cup \{\infty\}$ with $f(1) = 0$, take, for any $n \in \bbN$ and $\p \in \cP(n)$,
\begin{align}
	g(\p) = \frac{1}{n} \sum_{x=1}^n f(n \p_x) \,.
\end{align}
Property~1) is now satisfied because $g$ is symmetric and convex, and thus Schur convex. Convexity also implies continuity in the interior; however, we need to additionally assume here that the function is also continuous at the boundary. Property~2) is satisfied by assumption on $f$ and Property~3) can be verified by close inspection. More generally, for any divergence $\D$ that is continuous in the first argument, the function $g_{\D}(\p) := \D(\p \| \uu{n})$ for all $\p \in \cP(n)$ is a valid $g$-function. To verify this, note that 
\begin{align}
	\D(\p \| \uu{n}) = \D\big(\p \otimes \uu{k} \big\| \uu{nk} \big)
\end{align}
as a consequence of the DPI applied twice for channels introducing and removing an independent distribution $\uu{k}$.

\begin{theorem}
	Let $g$ be a function satisfying Properties~1)--3). For any $n \in \bbN$, $\p \in \cP(n)$ and $\q \in \cPp(n) \cap \bbQ^n$, we define
	\begin{align}
		D_g(\p\|\q) := g \left( \bigoplus_{x=1}^n p_x \uu{k_x} \right) ,  \label{eq:Dg}
	\end{align}
	where $\q = ( \frac{k_1}{k}, \frac{k_2}{k}, \ldots )$ for $k_i \in \bbN$, $k \in \bbN$. For general $\q \in \cP(n)$, $D_g$ is defined via continuous extension. Then, $D_g$ is a divergence and continuous in $\q$ for any fixed $\p \in \cP(n)$.
\end{theorem}
Alternatively, we can also write $D_g(\p\|\q) = g(\r)$ where $\r$ is constructed in Lemma~\ref{lem:embed}.

\begin{proof}
We first need to verify that $D_g$ is well-defined for $\q \in \cPp(n) \cap \bbQ^n$. Note that there is a freedom in choosing $k$ in Eq.~\eqref{eq:Dg}; however, due to Property~3 this does not change the value of $D_g$, and we can  pick the least common denominator.

Next we show that $D_g$ is indeed a divergence on this restricted space. First, note that normalisation holds since $D_g(1\|1) = g(1) = 0$. To show the DPI, let $\p \in \cP(n)$ and $W \in \cS(n,m)$ be a channel that maps rationals to rationals. Let further $k\in\bbN$ be large enough such that we can express
\begin{align}
	\q= \left(\frac{k_1}k, \ldots, \frac{k_n}k\right), \quad \q' = \q W =\left(\frac{k_1'}k, \ldots ,\frac{k_n'}k\right), \label{eq:qq}
\end{align}
where $k_x,k_x'\in\bbN$ with $\sum_{x=1}^nk_x=\sum_{x=1}^{n}k_x'=k$.
Let now $\r,\s\in \cP(k)$ be such that $(\p,\q) \sim (\r,\uu{k})$ and $(\p W, \q W) \sim (\s,\uu{k})$, where $\r$ and $\s$ are constructed using Lemma~\ref{lem:embed}. By definition,
$(\p,\q) \majo (\p W, \q W)$ so that $(\r,\uu{k}) \majo (\s,\uu{k})$, or, equivalently, $\r \majo \s$. Hence,
\begin{align}
D_g(\p W \| \q W) = g(\s) \leq g(\r) = D_{g}(\p\|\q)\;,
\end{align}
where the inequality follows from the Schur convexity of $g$.

Next we need to show continuity in the second argument on $\cPp(n) \cap \bbQ^n$, so that the continuous extension is well-defined. Given $\q$ and $\q'$ we introduce the channel $W$ given by
\begin{align}
	\ee{i} W = (1-\eps) \ee{i} + \eps \q + (\q' - \q)  \,,
\end{align}
where $\eps \geq 1 - 2^{-D_{\max}(\q'\| \q)}$ is chosen rational and can be arbitrarily small when $\q$ and $\q'$ approach each other. This choice of $\eps$ ensures that this is indeed a channel. Similarly, we define $W'$ and $\eps'$ with $\q$ and $\q'$ interchanged. We then have $\q W = \q'$ and $\q' W' = \q$, and furthermore,
\begin{align}
	\| \p - \p W \| = \big\| \eps( \p - \q ) + \q' - \q \big\| \leq \eps + \| \q - \q' \|,
\end{align}
and similarly $\| \p - \p W' \| \leq \eps' + \| \q - \q' \|$. Thus, $\p W \to \p$ and $\p W' \to \p$ as $\q \to \q'$.
Using the DPI we can thus bound
\begin{align}
	D_g(\p\|\q) - D_g(\p\|\q') &\geq D_g(\p W \|\q') - D_g(\p\|\q') \,,\\
	D_g(\p\|\q) - D_g(\p\|\q') &\leq D_g(\p \|\q) - D_g(\p W'\| \q) \,.
\end{align}
We now simply argue that since the two expressions on the right-hand side vanish when $\q \to \q'$ due to the continuity of $D_g$ in the first argument, which itself is inherited directly from the continuity of $g$, we have established continuity in the second argument.

Finally, we note that the DPI remains valid when we define the quantity for irrational $\q$ via continuous extension.
\end{proof}

\section{Properties of R\'enyi relative entropies}
\label{app:renyi}

\begin{lemma}
\label{lm:renyi2nd}
	Let $\vec{v} \in \mathbb{R}^d$ be such that $\sum_{x=1}^d v_x = 0$ und $\vec{u} = \uu{d}$. Then,
	for $\alpha \in [0, \infty]$, we have
	\begin{align}
   		\frac{\partial^2}{\partial \eps^2} D_{\alpha}(\vec{u} + \eps \vec{v} \| \vec{u} ) \Big|_{\eps=0} = \alpha d | \vec{v} |^2 ,
	\end{align}
	where $|\cdot|$ denotes the Euclidian vector norm.
\end{lemma}

\begin{proof}
  Consider first $\alpha \in [0, \infty) \setminus \{1\}$. We first note that
  \begin{align}
  	&\frac{\partial^2}{\partial \eps^2} D_{\alpha}(\vec{u} + \eps \vec{v} \| \vec{u}) \Big|_{\eps=0} \notag\\
	&\quad= \lim_{\eps \to 0} \left\{ \frac{D_{\alpha}(\vec{u} + \eps \vec{v} \| \vec{u}) + D_{\alpha}(\vec{u} -\eps \vec{v} \| \vec{u})}{\eps^2} \right\} \label{eq:def2nd} \\
	&\quad= \frac{1}{\alpha-1} \lim_{\eps \to 0} \left\{ \frac{\log q_+(\eps)}{\eps^2} +  \frac{\log q_-(\eps)}{\eps^2}  \right\} ,
  \end{align}
  where $q_{\pm}(\eps) = \sum_x u_x (1 \pm \eps d v_x)^{\alpha}$ and in the first equality we used the fact that $D_{\alpha}(\vec{u}\|\vec{u}) = 0$. Together with their derivatives $q_{\pm}'$ and $q_{\pm}''$, the functions satisfy $q_{\pm}(0) = 1$, $q_{\pm}'(0) = 0$ and $q_{\pm}''(0) = \alpha(\alpha-1) d | v |^2$.
  Using de l'H{\^o}pital's rule twice, the calculation proceeds as
 \begin{align}
  	&\frac{\partial^2}{\partial \eps^2} D_{\alpha}(\vec{u} + \eps \vec{v} \| \vec{u}) \Big|_{\eps=0} \\
	&\quad= \frac{1}{\alpha-1} \lim_{\eps \to 0} \left\{ \frac{q_+'(\eps)}{2 \eps q_+(\eps)} + \frac{q_-'(\eps)}{2 \eps q_-(\eps)} \right\} \\
	&\quad= \frac{1}{\alpha-1} \lim_{\eps \to 0} \left\{ \frac{q_+''(\eps)}{2 q_+(\eps) + 2\eps q_+'(\eps)} + \frac{q_-''(\eps)}{2 q_-(\eps) + 2 \eps q_-'(\eps)} \right\} \\
	&\quad= \alpha d |v|^2 \,.
  \end{align}
Finally, we note that for $\alpha = \infty$, the limit in Eq.~\eqref{eq:def2nd} diverges to $+\infty$, as required. A similar computation reveals that the statement of the lemma also holds for $\alpha = 1$.
\end{proof}


\section{Simpler catalytic relative majorisation}
\label{app:restrict}

\newcommand{\tq}{\tilde{\q}}

\begin{lemma}\label{oneside}
Let $m, n \in \bbN$, $\p, \q \in \cP(m)$ and $\p', \q' \in \cP(n)$. Suppose further that $\p'$ and $\q$ have full support. Then, $(\p, \q) \cata (\p', \q')$ if and only if there exist sequences $\{ \q_k \}_{k \in \bbN}$ and $\{ \q_k' \}_{k \in \bbN}$ with $\q_k\in \cP(m)$ and $\q_k'\in \cP(n)$ such that
$\q_k \to \q$ and $\q_k' \to \q'$ as $k \to \infty$, and $(\p, \q_k) \trump (\p', \q'_k)$ for all $k \in\bbN$.
\end{lemma}

\begin{proof}
Clearly, if the two sequences defined above exist then $(\p, \q) \cata (\p', \q')$ since one can define $\p_k\eqdef\p$, $\p'_k\eqdef\p'$ for all $k\in\bbN$. It thus remains to show the converse implication.

Suppose $(\p, \q) \cata (\p', \q')$ and let $\{\p_k\}_{k \in \bbN}$, $\{ \q_k \}_{k \in \bbN}$, $\{\p_k'\}_{k \in \bbN}$, and $\{ \q_k' \}_{k \in \bbN}$ be as in Definition~\ref{def:cata}. In particular, $(\p_k, \q_k) \trump (\p'_k, \q'_k)$ for all $k\in\bbN$. Now, define a channel $W$ by its action
\begin{align}
\ee{i} W=\p_k+2^{-D_{\max}(\p\|\p_k)}(\ee{i}-\p)\;,\quad\forall\; i \in [n] \;.
\end{align}
It is simple to check that $W$ is indeed a channel.
Define also 
\begin{align}
\tq_k\eqdef\p+2^{D_{\max}(\p\|\p_k)}(\q_k-\p_k)\;.
\end{align}
Note that since $D_{\max}(\p\|\p_k)\to 0$ as $k\to\infty$ and since we assume that $\q>0$ (and $\q_k\to\q$ as $k\to\infty$) we conclude that for sufficiently large $k\in\bbN$, $\tq_k\in\cP(m)$. Further,  observe that
\begin{align}
\p W=\p_k\quad,\quad\tq_k W=\q_k\;.
\end{align}
Therefore, $(\p, \tq_k)\majo(\p_k, \q_k)$ and $\tq_k\to\q$ as $k\to\infty$.

Next, let 
\begin{align}
\tq_k'\eqdef\p'+2^{-D_{\max}(\p_k'\|\p')}(\q_k'-\p_k')\;.
\end{align}
and note that $\tq_k'\in\cP(n)$. Further, define a channel $\tilde{W}$ by its action
\begin{align}
\ee{i} \tilde{W}=\p'+2^{-D_{\max}(\p_k'\|\p')}(\ee{i}-\p_k')\;,\quad\forall\;i \in [n]\;.
\end{align}
It is simple to check that $\tilde{W}$ is a channel and
\begin{align}
\p_k' \tilde{W}=\p'\quad,\quad\q_k' \tilde{W}=\tq_k\;.
\end{align}
Therefore, $(\p_k', \q_k')\majo(\p', \tq_k')$ and $\tq_k'\to\q'$ as $k\to\infty$.
With the above two constructions of $\{\tq_k\}$ and $\{\tq_k'\}$ we have
\begin{align}
(\p, \tq_k)\majo(\p_k, \q_k)\trump(\p'_k, \q'_k)\majo(\p', \tq_k')
\end{align}
where $\tq_k \to \q$ and $\tq_k' \to \q'$ as $k \to \infty$. 
\end{proof}


\section{Lemmas used in the proof of Theorem~\ref{thm:char2}}
\label{app:lemmas}

\begin{lemma}\label{lem1}
Let $\p,\q\in\cP(n)$ and suppose $\p\neq \q$. Then, there exists sequence $\{ \q_k \}_{k \in \bbN}$ with $\q_k\in \cPp(n) \cap \bbQ^n$ such that
$\q_k \to \q$ as $k \to \infty$ and for all $k\in\bbN$ 
\begin{align}
(\p,\q)\majo(\p,\q_k)\;.
\end{align}
\end{lemma}

\begin{proof}
W.l.o.g. we assume that
\begin{align}
\label{order}
\frac{p_1}{q_1}\geq\frac{p_2}{q_2}\geq\cdots \geq\frac{p_n}{q_n}\geq 0\;.
\end{align}
so that the vertices of $\cL(\p,\q)$ are given by $(a_k,b_k)$ where $a_k=\sum_{x=1}^kp_x$ and $b_k=\sum_{k=1}^kq_x$.
Note that the above relation implies that $q_n>0$. Let $\eps_1,...,\eps_{n-1}$ be small enough positive numbers such that for all $x=1,...,n-1$, $q_x'\eqdef q_x+\eps_x$ is a rational number. Furthermore, we can always choose $\eps_1,...,\eps_{n-1}$ to be small enough such that their sum $\eps\eqdef\sum_{x=1}^{n-1}\eps_x$ satisfies
$q_n'\eqdef q_n-\eps>0$. Note also that for these choices $\q'\eqdef(q_1',...,q_n')$ has positive rational numbers. 
Note that if for some $x\in[n]$, $\frac{p_x}{q_x}>\frac{p_{x-1}}{q_{x-1}}$ then for small enough $\eps>0$ also $\frac{p_x}{q_x'}>\frac{p_{x-1}}{q_{x-1}'}$. On the other hand, if for some $x\in[n]$, $\frac{p_x}{q_x}=\frac{p_{x-1}}{q_{x-1}}$ then if necessary we exchange between $(p_x,q_x)$ and $(p_{x-1},q_{x-1})$ so that $\frac{p_x}{q_x'}\geq\frac{p_{x-1}}{q_{x-1}'}$ still holds. In this way we can assume w.l.o.g. that both Eq.~\eqref{order} holds and
\begin{align}\label{order2}
\frac{p_1}{q_1'}\geq\frac{p_2}{q_2'}\geq\cdots \geq\frac{p_n}{q_n'}\;.
\end{align}
By construction, except for the extreme vertices $(0,0)$ and $(1,1)$, all the vertices of $\cL(\p,\q')$ are strictly above the vertices of $\cL(\p,\q)$; explicitly, note that for all $k=1,...,n-1$
\begin{align}
b_k'\eqdef\sum_{x=1}^kq_x'=b_k+\sum_{x=1}^k\eps_x>b_k\;.
\end{align}
Hence, $\cL(\p,\q')$ is everywhere above $\cL(\p,\q)$ so that $(\p,\q)\majo(\p,\q')$. Finally, since $\{\eps_x\}$ can be made arbitrarily small, we can construct in this way a sequence $\{ \q_k \}_{k \in \bbN}$ with the desired properties.
\end{proof}

\begin{lemma}\label{lem2}
Let $\p,\q\in\cP(n)$. Then, there exists sequence $\{ \q_k \}_{k \in \bbN}$ with $\q_k\in \cP(n) \cap \bbQ^n$ such that
$\q_k \to \q$ as $k \to \infty$, $\supp(\q_k)=\supp(\q)$, and for all $k\in\bbN$
\begin{align}
(\p,\q_k)\majo (\p,\q)\;.
\end{align}
\end{lemma}
\begin{proof}
Similarly to the previous lemma, we assume w.l.o.g. that Eq.~\eqref{order} holds. Let $r\in[n-1]$ be such that $q_1=\cdots=q_r=0$ and $q_{r+1}>0$. Define $q'_x\eqdef 0$ for $x\in[r]$. For each $x=r+1,...,n-1$ define $\eps_x$ to be small enough positive real numbers such that $q_x'\eqdef q_x-\eps_x$ are positive rational numbers. Further, define $q_n'\eqdef 1-\sum_{x=r+1}^{n-1}q_x'$.
Hence, by construction, $\q'\eqdef(q_1,...,q_n)$ is a vector with rational components with the same support as $\q$. Moreover, following the same arguments as in the previous Lemma we can assume w.l.o.g. that also in this case Eq.~\eqref{order2} holds for small enough $\{\epsilon_x\}_{x=r+1}^{n-1}$. Hence, for any $k=1,...,n$, the $k^{\text{th}}$-vertex $(a_k,b_k)$ of $\cL(\p,\q)$ is not below the $k^{\text{th}}$-vertex $(a_k,b_k')$ of $\cL(\p,\q')$. That is, $(\p,\q')\majo (\p,\q)$, and since the $\{\epsilon_x\}$ can be made arbitrarily small we can construct in this way a sequence $\{ \q_k \}_{k \in \bbN}$ with the desired properties.
\end{proof}

\end{document}